\documentclass{article}
\usepackage{geometry}
\usepackage[nospace,noadjust]{cite}
\usepackage{amsmath,amssymb,amsfonts,times}
\usepackage{graphicx}
\usepackage{textcomp}
\usepackage{xcolor}
\usepackage{subfigure}
\usepackage{epsfig}
\usepackage{multirow}
\usepackage{algorithm}
\usepackage{algorithmicx}
\usepackage{algpseudocode}
\usepackage{array}
\usepackage{epstopdf}
\usepackage{color}
\usepackage{diagbox}
\usepackage{bm}
\usepackage{tcolorbox}
\usepackage{pdfsync}

\usepackage{url}
\usepackage[hidelinks]{hyperref}
\usepackage{amsmath,amsthm,amssymb,amsbsy}
\usepackage{paralist}
\usepackage{xcolor}
\usepackage{color}
\usepackage{graphicx}
\graphicspath{{./figs/}}
\usepackage{comment}
\usepackage{multirow}
\usepackage[toc, page]{appendix}

\usepackage{graphicx}
\usepackage{fancyhdr}
\usepackage{cite}
\usepackage{cleveref}

\newtheorem{lemma}{Lemma}
\newtheorem{definition}{Definition}

\newtheorem{theorem}{Theorem}

\theoremstyle{remark}

\theoremstyle{problem}

\newcommand{\R}{\mathbb{R}}
\def \real    { \mathbb{R} }

\newcommand{\C}{\mathbb{C}}

\newcommand{\e}{\begin{equation}}
\newcommand{\ee}{\end{equation}}
\newcommand{\en}{\begin{equation*}}
\newcommand{\een}{\end{equation*}}
\newcommand{\eqn}{\begin{eqnarray}}
\newcommand{\eeqn}{\end{eqnarray}}
\newcommand{\bmat}{\begin{bmatrix}}
\newcommand{\emat}{\end{bmatrix}}
\renewcommand{\Re}[1]{\operatorname{Re}\left\{#1\right\}}
\renewcommand{\Im}[1]{\operatorname{Im}\left\{#1\right\}}

\DeclareMathAlphabet\mathbfcal{OMS}{cmsy}{b}{n}
\renewcommand{\P}[1]{\operatorname{\mathbb{P}}\left(#1\right)}
\newcommand{\E}{\operatorname{\mathbb{E}}}

\newcommand{\vct}[1]{\boldsymbol{#1}}
\newcommand{\mtx}[1]{\boldsymbol{#1}}

\newcommand{\<}{\langle}
\renewcommand{\>}{\rangle}
\renewcommand{\H}{\mathrm{H}}

\newcommand{\trace}{\operatorname{trace}}

\newcommand{\rank}{\operatorname{rank}}

\newcommand{\dist}{\operatorname{dist}}

\newcommand{\set}[1]{\mathbb{#1}}

\DeclareMathOperator*{\minimize}{\text{minimize}}

\DeclareMathOperator*{\argmin}{\text{arg~min}}

\def \st {\operatorname*{s.t.\ }}

\newcommand{\wh}{\widehat}

\newcommand{\norm}[2]{\left\| #1 \right\|_{#2}}

\newcommand{\abs}[1]{\left| #1 \right|}
\newcommand{\bracket}[1]{\left( #1 \right)}
\newcommand{\parans}[1]{\left(#1\right)}
\newcommand{\innerprod}[2]{\left\langle #1,  #2 \right\rangle}

\newcommand{\calA}{\mathcal{A}}

\newcommand{\calP}{\mathcal{P}}

\newcommand{\calU}{\mathcal{U}}

\newcommand{\va}{\vct{a}}

\newcommand{\vc}{\vct{c}}

\newcommand{\ve}{\vct{e}}

\newcommand{\vp}{\vct{p}}

\newcommand{\vx}{\vct{x}}
\newcommand{\vy}{\vct{y}}

\newcommand{\vrho}{\vct{\rho}}
\newcommand{\vzero}{\vct{0}}

\newcommand{\mA}{\mtx{A}}
\newcommand{\mB}{\mtx{B}}
\newcommand{\mC}{\mtx{C}}
\newcommand{\mD}{\mtx{D}}
\newcommand{\mE}{\mtx{E}}

\newcommand{\mP}{\mtx{P}}
\newcommand{\mQ}{\mtx{Q}}
\newcommand{\mR}{\mtx{R}}

\newcommand{\mU}{\mtx{U}}
\newcommand{\mV}{\mtx{V}}
\newcommand{\mW}{\mtx{W}}
\newcommand{\mX}{\mtx{X}}
\newcommand{\mY}{\mtx{Y}}
\newcommand{\mZ}{\mtx{Z}}

\newcommand{\mDelta}{\mtx{\Delta}}
\newcommand{\mLambda}{\mtx{\Lambda}}

\newcommand{\mSigma}{\mtx{\Sigma}}

\newcommand{\mId}{{\bf I}}

\newcommand{\setS}{\set{S}}

\graphicspath{{./figs/}}

\newlength{\imgwidth}
\newboolean{twoColVersion}
\setboolean{twoColVersion}{false}
\newcommand{\twoCol}[2]{\ifthenelse{\boolean{twoColVersion}} {#1} {#2} }

\newtheorem{Corollary}{Corollary}

\title{Optimal Allocation  of Pauli Measurements \\ for Low-rank  Quantum State Tomography}

\author{Zhen Qin, Casey Jameson, Zhexuan Gong,  Michael B. Wakin and Zhihui Zhu\thanks{ZQ (email: qin.660@osu.edu) and ZZ (email: zhu.3440@osu.edu)  are with the Department of Computer Science and Engineering, the Ohio State University; CJ (email: cwjameson@mines.edu) and ZG (email: gong@mines.edu) are with the Department of Physics, Colorado School of Mines; and MBW (email: mwakin@mines.edu) is with the Department of Electrical Engineering, Colorado School of Mines.}
}

\begin{document}

\maketitle

\begin{abstract}
The process of reconstructing quantum states from experimental measurements, accomplished through quantum state tomography (QST), plays a crucial role in verifying and benchmarking quantum devices.  A key challenge of QST is to find out how the accuracy of the reconstruction depends on the number of state copies used in the measurements. When multiple measurement settings are used, the total number of state copies is determined by multiplying the number of measurement settings with the number of repeated measurements for each setting. Due to statistical noise intrinsic to quantum measurements, a large number of repeated measurements is often used in practice. However, recent studies have shown that even with single-sample measurements--where only one measurement sample is obtained for each measurement setting--high accuracy QST can still be achieved with a sufficiently large number of different measurement settings. In this paper, we establish a theoretical understanding of the trade-off between the number of measurement settings and the number of repeated measurements per setting in QST. Our focus is primarily on low-rank density matrix recovery using Pauli measurements. We delve into the global landscape underlying the low-rank QST problem and demonstrate that the joint consideration of measurement settings and repeated measurements ensures a bounded recovery error for all second-order critical points, to which optimization algorithms tend to converge. This finding suggests the advantage of minimizing the number of repeated measurements per setting when the total number of state copies is held fixed. Additionally, we prove that the Wirtinger gradient descent algorithm can converge to the region of second-order critical points with a linear convergence rate. We have also performed numerical experiments to support our theoretical findings.
\end{abstract}

\section{Introduction}
\label{intro}

In the realm of quantum information processing that encompasses quantum computation \cite{nielsen2002quantum, wang2012quantum}, quantum communication \cite{gisin2007quantum, cozzolino2019high}, quantum metrology \cite{braginsky1995quantum, brukner2017quantum}, and more, a pivotal task lies in acquiring knowledge about quantum states. The best way of achieving said task is known as quantum state tomography (QST) \cite{bertrand1987tomographic,vogel1989determination,leonhardt1995quantum,hradil1997quantum,qin2024quantum}, which allows one to fully reconstruct the experimental quantum state with high accuracy. For an $n$-qubit quantum system, its quantum state is represented by a density matrix $\vrho$ of dimension $2^n\times 2^n$. To estimate an unknown experimental quantum state $\vrho$, one typically needs to perform quantum measurements on multiple identical copies of the state. In the most general scenario, these quantum measurements are described by the so-called Positive Operator-Valued Measures (POVMs), which are collections of positive semi-definite (PSD) matrices or operators $\{\mE_1,\ldots,\mE_Q\}$ that sum up to the identity operator. Each operator $\mE_q$ ($q=1,\dots, Q$) in a POVM corresponds to a potential measurement outcome that occurs with a probability denoted by $p_q = \trace(\mE_q \vrho)$. Due to the inherent probabilistic nature of quantum measurements, a precise estimation of $\vrho$ usually requires repeating the measurements a considerable number of times (say $M$) using one or more sets of POVMs. This repetition allows one to obtain a statistically accurate estimate $\wh p_q$ of each $p_q$.

An important question in QST is how many state copies are needed for the measurements to reconstruct a generic quantum state within a given error threshold. The answer to this question depends critically on the choice of measurement settings (i.e., the POVM) used. Commonly studied measurement settings include Pauli measurements \cite{liu2011universal, guctua2020fast}, Clifford measurements \cite{francca2021fast}, and Haar-distributed unitary measurements \cite{voroninski2013quantum}, 4-design \cite{KuengACHA17}, etc. Unless the measurement setting provides information completeness, more than one measurement setting needs to be used. A simple example is the Pauli basis measurement used widely in various quantum experiments, where each qubit is measured in the Pauli $x,y,z$ bases respectively, resulting in $3^n$ different measurement bases/settings necessary for obtaining full information of the quantum state. In general, we shall denote $K$ as the number of measurement settings, and $M$ as the number of measurements for each setting. The total number of state copies used for QST is thus given by $KM$.
In most of the previous studies, $M$ is chosen to be large such that the empirical probabilities $\wh p_q$ are close to the true probabilities $p_q$. Conversely, if $M$ is small, the signal-to-noise ratio between the clean measurements $\{p_q\}$ and the statistical errors $\{\wh p_q - p_q\}$ could become substantially less than $1$, making accurate QST challenging. Surprisingly, recent works \cite{wang2020scalable, huang2020predicting} show that as long as $KM$ is large, even setting $M=1$ can lead to high-accuracy QST, which is somewhat counter-intuitive as each measurement gives little information of the true probability distribution for the particular measurement setting. When the total number of state copies $KM$ is held fixed, it appears that the use of a large number of different measurement settings $K$ can effectively compensate for the large statistical error associated with a low number of measurements per setting $M$. Consequently, the primary question we address in this paper is:

\smallskip
\begin{tcolorbox}[colback=white,left=1mm,top=1mm,bottom=1mm,right=1mm,boxrule=.3pt]
\centering {\bf Question}: What is the trade-off between the number of measurement settings ($K$)\\ and the number of repeated measurements per setting ($M$) in QST?
\end{tcolorbox}

\paragraph{Main results}
Our objective is to answer the main question posed above with a focus on reconstructing a quantum state represented by a low-rank density matrix using Pauli measurements, a topic that has been studied extensively \cite{liu2011universal,flammia2012quantum,kim2023fast,hsu2024quantum}. Specifically, for an $n$-qubit system, we reconstruct the target state from the experimental measurements of a set of $n$-body Pauli observables $\{\mW_k\}_{k=1}^{K}$. Each $\mW_k$ is represented by a tensor product of single-qubit Pauli matrices $\sigma_1$, $\sigma_2$, and $\sigma_3$, as well as the identity matrix $\sigma_0$ (see \eqref{Pauli basis in 2x2} for their explicit forms): $\mW_k = \sigma_{k_1} \otimes \sigma_{k_2}\cdots \sigma_{k_n}$, where $k_i=0,1,2,3$ for $i=1,2,\cdots,n$. Physically, each $\mW_k$ can be measured by measuring $\sigma_{k_i}$ for the $i$-th qubit for $i=1,2,\cdots,n$. Such measurement is routinely performed in most current quantum computing hardware platforms. In total, there are $4^n$ distinct $n$-body Pauli observables denoted by $\{\mW_i\}_{i=1}^{4^n}$. Our QST protocol randomly chooses $K$ of them with replacement and measures each of them $M$ times, thus consuming $KM$ total state copies. It is noteworthy that   some $\{\mW_k\}$ may be measured more than once. While for a generic quantum state, one needs to measure all elements of $\{\mW_i\}_{i=1}^{4^n}$ to gain full information of the state needed for accurate QST, for a low-rank density matrix, measuring only a small random set of $\{\mW_k\}_{k=1}^{K}$ is sufficient for recovering the target state with high precision \cite{liu2011universal,flammia2012quantum}.

Our main result is to show that for a fixed number of total state copies given by $KM$, choosing a larger $K$ (for a sufficiently large $K$) always leads to a smaller upper bound of the recovery error. Our numerical simulation results with random low-rank physical states further show that the actual recovery error also decreases as $K$ is increased, with $M=1$ generally giving the smallest recovery error. This suggests that the optimal Pauli observable measurement strategy is to maximize of the number of measurement settings. Mathematically, we can state our main result with the following theorem:

\begin{theorem} (informal version of \Cref{upper bound recovery error of QST for Pauli measurements} and Corollary~\ref{LocalConvergence_Conclusion_1})
\label{informal combination of two theorems}
Given an $n$-qubit density matrix $\vrho$ of rank $r$, randomly select $K$ $n$-body Pauli observables from the full set $\{\mW_i\}_{i=1}^{4^n}$ and perform measurements of each $M$ times.
For any $\epsilon>0$, if $K \geq \Omega(2^n r n^6)$, $M \leq O(K/n)$, and $KM \geq \Omega(g(K) 4^n n r/\epsilon^2)$, where $g(K)$ is $\Omega(1)$ but decreases monotonically as $K$ increases, then with high probability, every local minimum of a particular low-rank least squares estimator that involves the measurement results can recover the density matrix with an $\epsilon$-closeness in the Frobenius norm, and such minima can be achieved by employing a Wirtinger gradient descent algorithm with a good initialization.
\end{theorem}

Our results quantify the impact of the total number of state copies $KM$ on the recovery error. Importantly, since the function $g(K)$ and the failure probability are both inversely related to $K$, increasing $K$ while keeping $KM$ constant can substantially decrease the recovery error while boosting the success probability.
In addition, we propose a Wirtinger gradient-based factorization approach, which proves the existence of a method that can converge to the region of favorable local minima with a linear convergence rate when appropriately initialized. Our theoretical findings are substantiated by numerous experimental results.

\paragraph{Related works}

Numerous theoretical studies in the QST literature~\cite{flammia2012quantum,KuengACHA17,zhang2017efficient,kyrillidis2018provable,brandao2020fast,guctua2020fast,francca2021fast} have proposed methodologies involving the least-squares loss function to recover low-rank density matrices. It is noteworthy that \cite{kyrillidis2018provable,guctua2020fast} primarily investigate error and algorithmic aspects utilizing Pauli measurements. Drawing upon the projected least squares method and employing Pauli measurements, which include both Pauli observable and Pauli basis measurements, \cite{guctua2020fast}~examined the necessary number of total state copies ($KM$) to achieve accurate recovery in the trace norm. However, the trade-off between $K$ and $M$ was not explored. While \cite{kyrillidis2018provable} dealt with a loss function similar to ours, it did not address measurement noise, which arises due to the utilization of empirical Pauli measurements. Moreover, the algorithm analysis failed to establish the connection between recovery error and the total number of state copies ($KM$).

The QST problem can be viewed as complex matrix sensing, involving a specific type of measurement operator and inherently probabilistic measurements. In the last few decades, nonconvex optimization for matrix sensing has evolved in two main directions: (1) local geometry of two-stage algorithms \cite{jain2013low,zheng2015convergent,Tu16,charisopoulos2021low,tong2021accelerating,li2020nonconvex}, which rely on optimization algorithms with proper initialization; and (2) global landscape analysis and initialization-free algorithms \cite{bhojanapalli2016global,park2017non,wang2017unified,ge2017no,zhu2018global,li2019non,zhu2021global}. However, these prior works have predominantly focused on real matrix sensing in noiseless environments or with Gaussian noise, making them unsuitable for addressing QST involving complex matrices and measurement noise. Note that as a special case of matrix sensing, previous studies have explored the local geometry and global landscape of complex phase retrieval with a rank-one target matrix \cite{candes2015phase, sun2018geometric}. To solve this problem efficiently, the Wirtinger flow algorithm with a convergence guarantee has been provided in \cite{candes2015phase}. We consider an algorithm that is essentially the same as Wirtinger flow, but we highlight two important differences. First, we employ a different step size than that in~\cite{candes2015phase}, where it was prescribed to be inversely proportional to the magnitude of the initial guess. Second, whereas \cite{candes2015phase} studied the recovery of a rank-one matrix from rank-one measurements, our theoretical results extend these findings to the case of rank-$r$ complex matrix recovery with measurement noise.

\subsection{Notation}
\label{sec:notation}

We use bold capital letters (e.g., $\mX$) to denote matrices,  bold lowercase letters (e.g., $\vx$) to denote vectors, and italic letters (e.g., $x$) to denote scalar quantities. The calligraphic letter $\calA$ is reserved for the linear measurement map.  For a positive integer $K$, $[K]$ denotes the set $\{1,\dots, K \}$. The superscripts $(\cdot)^\top$ and $(\cdot)^{\rm H}$ denote the transpose and Hermitian transpose operators, respectively, while the superscript $(\cdot)^*$ denotes the complex conjugate. For two matrices $\mA,\mB$ of the same size, $\innerprod{\mA}{\mB} = \trace(\mA^{\rm H}\mB)$ denotes the inner product between them.
$\|\mA\|$ and $\|\mA\|_F$ respectively represent the spectral norm and Frobenius norm of $\mA$.
For a vector $\va$ of size $N\times 1$, its $l_n$-norm is defined as $||\va||_n=(\sum_{m=1}^{N}|a_m|^n)^\frac{1}{n}$.
For two positive quantities $a,b\in \real$, the inequality $b\lesssim a$ or $b = O(a)$ means $b\leq c a$ for some universal constant $c$; likewise, $b\gtrsim a$ or $b = \Omega(a)$ represents $b\ge ca$ for some universal constant $c$. $\calU_r$ is the set of unitary matrices of size $r$.

\section{Preliminary}
\label{Preliminary}
In this section, we provide a brief review of QST based on Pauli measurements, starting from basic concepts in quantum mechanics such that quantum state and quantum measurement.

\paragraph{Quantum state and measurements}
For a quantum system composed of $n$ qubits, its quantum state is represented by a density matrix $
\vrho\in\C^{2^n\times 2^n}$ that obeys two properties: $(i)$~ $\vrho \succeq \vzero$ is a positive semidefinite (PSD) matrix, and $(ii)$~$\trace(\vrho) = 1$. The objective of quantum state tomography is to construct or estimate the density matrix $\vrho$ of a physical quantum system by performing measurements on an ensemble of identical state copies.

The most general measurements on a quantum system are denoted by Positive Operator Valued Measures (POVMs)  \cite{nielsen2002quantum}. Specifically, a POVM is a set of PSD matrices
\begin{eqnarray}
\label{The defi of POVM}
\{\mE_1,\ldots,\mE_Q \}\in\C^{2^n\times 2^n}, \ \ \st  \sum_{q=1}^Q \mE_q = \mId_{2^n}.
\end{eqnarray}
Each POVM element $\mE_q$ is linked to a potential outcome of a quantum measurement. The probability $p_q$ of detecting the $q$-th outcome while measuring the density matrix $\vrho$ can be expressed as follows:
\begin{eqnarray}
\label{The defi of POVM 2}
p_q = \innerprod{\mE_q}{\vrho},
\end{eqnarray}
where $\sum_{q=1}^Qp_q=1$ due to \eqref{The defi of POVM} and the fact that $\trace(\vrho) = 1$. Given the probabilistic nature of quantum measurements, it is necessary to prepare multiple copies of the quantum state in the experiment. In other words, since it is not possible to directly obtain $\{p_q\}$ in a single physical experiment, one typically repeats the measurement process $M$ times. By taking the average of the statistically independent outcomes, one can obtain the empirical probabilities as follows:
\begin{equation}
\wh p_{q} = \frac{f_q}{M}, \  q \in[Q]:=\{1,\ldots,Q\},
\label{eq:empirical-prob}
\end{equation}
where $f_q$ denotes the number of times the $q$-th outcome is observed in the $M$ experiments. The empirical probabilities $\{\wh p_q\}$ can then be used to recover or estimate the unknown density operator $\vrho$.

Collectively, the random variables $f_1,\ldots,f_Q$ are characterized by the multinomial distribution $\operatorname{Multinomial}(M, \vp)$ \cite{severini2005elements} with parameters $M$ and $\vp = \begin{bmatrix} p_1 & \cdots & p_Q\end{bmatrix}^\top$, where $p_q$ is defined in \eqref{The defi of POVM 2}. It follows that the empirical probability $\wh p_q$ in \eqref{eq:empirical-prob} is an unbiased estimator of the true probability $p_q$.
For convenience, we call $\{p_q\}$ and $\{\wh p_q\}$ the population and empirical (linear) measurements, respectively.

\paragraph{Pauli observables}
Pauli observables/matrices are widely used in quantum measurements. For a single qubit, the Pauli observables are represented by the following $2\times2$ matrices:
\begin{equation}
\label{Pauli basis in 2x2}
\sigma_0 = \begin{bmatrix}1 & 0\\  0 & 1 \end{bmatrix},\ \sigma_1 = \sigma_x= \begin{bmatrix}0 & 1 \\ 1 & 0 \end{bmatrix},\ \sigma_2 = \sigma_y=\begin{bmatrix}0 & i \\ -i & 0 \end{bmatrix},\ \sigma_3 = \sigma_z=\begin{bmatrix}1 & 0 \\ 0 & -1 \end{bmatrix}.
\end{equation}
Here we include the identity matrix $\sigma_0$ which is technically not a Pauli matrix for the purpose of forming an orthonormal basis for $\C^{2\times 2}$ (with proper normalization). For $n$ qubits, letting $D = 2^n$, we can obtain $D$-dimensional Pauli matrices $\mW_k\in\C^{D\times D}$ by forming $n$-fold tensor products of $\sigma_0,\sigma_1,\sigma_2,\sigma_3$:
\begin{equation}
\mW_k =  \sigma_{k_1}\otimes \sigma_{k_2}\otimes \cdots \otimes \sigma_{k_n}, \ (k_1,\ldots,k_n)\in \{0,1,2,3\}^n, \ \forall k = 1, \ldots,D^2.
\label{eq:Pauli-matrix}
\end{equation}
Each $D$-dimensional Pauli matrix $\mW_k$ has full rank with eigenvalues $\pm 1$. Moreover, the properly normalized $D$-dimensional Pauli matrices, $\frac{1}{\sqrt{D}}\{\mW_1,\ldots,\mW_{D^2}\}$, form an orthonormal basis for $\C^{D\times D}$. Because of this, any density matrix $\vrho\in\C^{D\times D}$ is fully characterized by the inner products $\<\mW_1, \vrho \>, \dots, \<\mW_{D^2}, \vrho \> $, which are physically the expectation values of the $n$-qubit Pauli observable $\mW_k$ and are referred to as the {\it Pauli observable measurements}. The problem of QST---estimating $\vrho$---can therefore be formulated as the problem of estimating all $\{\<\mW_k, \vrho \>\}_{k=1}^{D^2}$. Experimentally, we can measure each $\mW_k$ using the following Pauli basis measurement that is routinely performed in various quantum computing platforms.

\paragraph{Pauli basis measurement}

Again let $D = 2^n$. Following \cite{flammia2012quantum,kim2023fast,hsu2024quantum},  we can expand each $\mW_k$ using linear combinations of local Pauli basis measurement operators.  Specifically, we can associate each two-dimensional Pauli basis matrix $\sigma_i$ with a two-outcome POVM $\{\mE_i^{\pm}\}$ corresponding to its eigenvectors. For example,
\begin{eqnarray}
\label{defi local Pauli measurement}
&&\sigma_1  = \begin{bmatrix}0 & 1 \\ 1 & 0 \end{bmatrix} = \underbrace{ \begin{bmatrix} \frac{1}{\sqrt{2}} \\ \frac{1}{\sqrt{2}}  \end{bmatrix}  \begin{bmatrix} \frac{1}{\sqrt{2}} & \frac{1}{\sqrt{2}}  \end{bmatrix} }_{\mE_1^+}
- \underbrace{ \begin{bmatrix} \frac{1}{\sqrt{2}} \\ - \frac{1}{\sqrt{2}}  \end{bmatrix}  \begin{bmatrix} \frac{1}{\sqrt{2}} & - \frac{1}{\sqrt{2}}  \end{bmatrix} }_{\mE_1^- },\\
&&\sigma_2  = \begin{bmatrix}0 & i \\ -i & 0 \end{bmatrix} = \underbrace{ \begin{bmatrix} \frac{1}{\sqrt{2}} \\ -\frac{i}{\sqrt{2}}  \end{bmatrix}  \begin{bmatrix} \frac{1}{\sqrt{2}} & \frac{i}{\sqrt{2}}  \end{bmatrix} }_{\mE_2^+}
- \underbrace{ \begin{bmatrix} \frac{1}{\sqrt{2}}\\  \frac{i}{\sqrt{2}}  \end{bmatrix}  \begin{bmatrix} \frac{1}{\sqrt{2}} &  -\frac{i}{\sqrt{2}}  \end{bmatrix} }_{\mE_2^- },\\
&&\sigma_3  = \begin{bmatrix}1 & 0 \\ 0 & -1 \end{bmatrix} = \underbrace{ \begin{bmatrix} 1 \\ 0  \end{bmatrix}  \begin{bmatrix} 1 & 0  \end{bmatrix} }_{\mE_3^+}
- \underbrace{ \begin{bmatrix} 0 \\ 1  \end{bmatrix}  \begin{bmatrix} 0 & 1  \end{bmatrix} }_{\mE_3^- }.
\end{eqnarray}
Here, each $\{\mE_i^{\pm}\}$ represents a rank-one orthonormal POVM\footnote{The measurement operators within a rank-one orthonormal POVM are projection operators onto orthonormal states of the Hilbert space} for $i = 1, 2, 3$. It is important to note that since $\sigma_0$ is an identity matrix, we can associate $\sigma_0$ with any of the rank-one orthonormal POVM elements $\{\mE_i^{\pm}\}$, specifically, $\sigma_0 = \mE_i^+ + \mE_i^-$ for all $i = 1, 2, 3$. To simplify the notation, we can set $\mE_0^\pm = \mE_3^\pm$. Thus, we can associate each Pauli basis matrix $\sigma_i$ with a two-outcome POVM $\{\mE_i^{\pm}\}$ in a unified way as
\begin{eqnarray}
\label{unified way of Pauli matrix sigma}
\sigma_i = \mE_i^+ + \wh\alpha_i \mE_i^-, \ \forall i = 0,1,2,3, \ \text{where} \ \wh\alpha_0 = 1, \wh\alpha_1 = \wh\alpha_2 = \wh\alpha_3 = -1.
\end{eqnarray}
Furthermore, we can rewrite the $k$-th $D$-dimensional Pauli matrix $\mW_k$ as
\begin{eqnarray}
\label{eq:Pauli-local-POVM}
\mW_k &\!\!\!\! =\!\!\!\!&  \sigma_{k_1}\otimes  \cdots \otimes \sigma_{k_n} = ( \mE_{k_1}^+  + \wh\alpha_{k_1} \mE_{k_1}^-) \otimes  \cdots \otimes (\mE_{k_n}^+  + \wh \alpha_{k_n} \mE_{k_n}^-) \nonumber\\
&\!\!\!\! =\!\!\!\!& \mE_{k_1}^+ \otimes \mE_{k_2}^+ \otimes \cdots \mE_{k_n}^+ + \wh\alpha_{k_1} \mE_{k_1}^- \otimes \mE_{k_2}^+ \otimes \cdots \otimes  \mE_{k_n}^+ \nonumber\\
&\!\!\!\! \!\!\!\!&  + \cdots + \wh\alpha_{k_1} \mE_{k_1}^-\otimes \wh\alpha_{k_2} \mE_{k_2}^- \otimes \cdots \otimes \wh \alpha_{k_n} \mE_{k_n}^- \nonumber\\
&\!\!\!\!=:\!\!\!\!& \sum_{j = 1}^D \alpha_{k,j} \mE_{k,j}, \quad \alpha_{k,1},\ldots, \alpha_{k,D} \in \pm 1,
\end{eqnarray}
which implies that we can define a Pauli basis-based $D$-outcome POVM:
\begin{eqnarray}
\label{defi of POVM in Pauli basis measurements}
\bigg\{\mE_{k,1},\ldots,\mE_{k,D}\bigg\} = \bigg\{\mE_{k_1}^{\pm} \otimes \cdots \otimes \mE_{k_n}^{\pm}\bigg\}.
\end{eqnarray}
Quantum measurements with this POVM will be characterized by $D$ probabilities (population measurements), which we refer to as the {\it Pauli basis measurements}:
\begin{eqnarray}
\label{defi of Pauli basis measurements}
\bigg\{ p_{k,1},\ldots,  p_{k,D}\bigg\} = \bigg\{\<\mE_{k,1}, \vrho\>, \ldots, \<\mE_{k,D}, \vrho\> \bigg\}.
\end{eqnarray}
Finally, the Pauli observable measurement with $\mW_k$ can be expressed using a linear combination of these probabilities:
\begin{eqnarray}
\label{association Pauli measurement and Pauli basis measurements}
\<\mW_k, \vrho \> = \sum_{j=1}^D \alpha_{k,j} \cdot  p_{k,j}.
\end{eqnarray}

From the discussion above, it follows that we can obtain empirical estimates of $\{\<\mW_k, \vrho \>\}$ by performing the Pauli basis measurements in \eqref{defi of Pauli basis measurements}. Specifically, for a fixed $k \in \{1,2,\dots,D^2\}$, we perform the experiments $M$ times and obtain the empirical probabilities $\{\wh p_{k,1},\ldots, \wh p_{k,D}\}$. Then, the $k$-th Pauli observable measurement
can be estimated as follows:
\begin{eqnarray}
\label{empirical Pauli basis measurements}
\<\mW_k, \vrho \> \approx \sum_{j=1}^D \alpha_{k,j} \cdot \wh p_{k,j},
\end{eqnarray}
where the right-hand side is also called an empirical Pauli observable measurement.

\paragraph{Statistical error in the empirical Pauli observables} Since the measurement of each Pauli observable $\mW_k$ can only be done experimentally for a finite number of times, the empirical estimate of each $\<\mW_k, \vrho \>$  is subject to statistical errors. Intuitively, one might expect that empirical Pauli observables
obtained from empirical Pauli basis measurements could exhibit large statistical errors, as they involve a great number
of empirical probabilities. However, according to \Cref{lem:MSE-empricial-pmfs} in {Appendix} \ref{sec: auxiliary materials}, the following result demonstrates that the statistical error for this case can be bounded by the order of $1/M$.
\begin{lemma}
For any Pauli matrix $\mW_k\in\C^{D\times D}$ as defined in \eqref{eq:Pauli-matrix}, the term $\sum_{j=1}^D \alpha_{k,j} \cdot \wh p_{k,j}$ in \eqref{empirical Pauli basis measurements} serves as an unbiased estimator for $\<\mW_k, \vrho \>$ with a statistical error bounded as:
\begin{align}
\E \left[\parans{\bigg(\sum_{j=1}^D \alpha_{k,j} \cdot \wh p_{k,j}\bigg) - \<\mW_k, \vrho \> }^2 \right]  \le \frac{2 - 2/D}{M} \le \frac{2}{M}.
\label{eq:MSE-empricial-Pauli}\end{align}
\label{lem:MSE-empricial-Pauli}\end{lemma}
Note that the left hand side of \eqref{eq:MSE-empricial-Pauli} is the variance of the error in estimating $\<\mW_k, \vrho \> $.
The result in \eqref{eq:MSE-empricial-Pauli} indicates that the statistical error does not increase significantly with respect to $D$. To support this argument, we conduct a numerical experiment as follows: We randomly generate $100$ Pauli matrices by selecting $100$ indices from the set $\{1,\dots, D^2 \}$ with replacement and a density matrix $\vrho = \mU \mU^{\rm H}\in\C^{2^n\times 2^n}$, where $\mU = \frac{ \mA + i \cdot \mB}{\|\mA + i \cdot \mB\|_F}\in \C^{2^n\times r}$ with the entries of $\mA$ and $\mB$ being independent and identically distributed (i.i.d.)\ from the standard normal distribution. In Figures~\ref{n = 1 error bound}-\ref{n = 7 error bound}, we calculate $\E [(\sum_{j=1}^D \alpha_{k,j} \cdot \wh p_{k,j} - \<\mW_k, \vrho \> )^2 ]$  (averaged over 100 Monte Carlo trials) for various numbers of qubits $n$ and experiments $M$ and show that it generally scales as $\frac{1}{M}$, consistent with \Cref{lem:MSE-empricial-Pauli}.

\begin{figure}[!ht]
\centering
\subfigure[]{
\begin{minipage}[t]{0.4\textwidth}
\centering
\includegraphics[width=6.5cm]{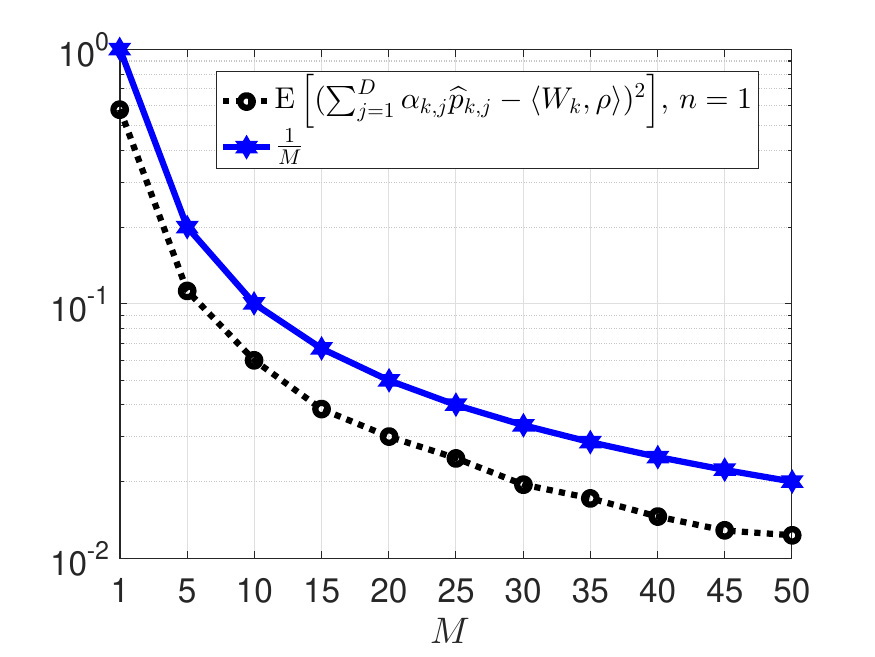}
\end{minipage}
\label{n = 1 error bound}
}
\subfigure[]{
\begin{minipage}[t]{0.4\textwidth}
\centering
\includegraphics[width=6.5cm]{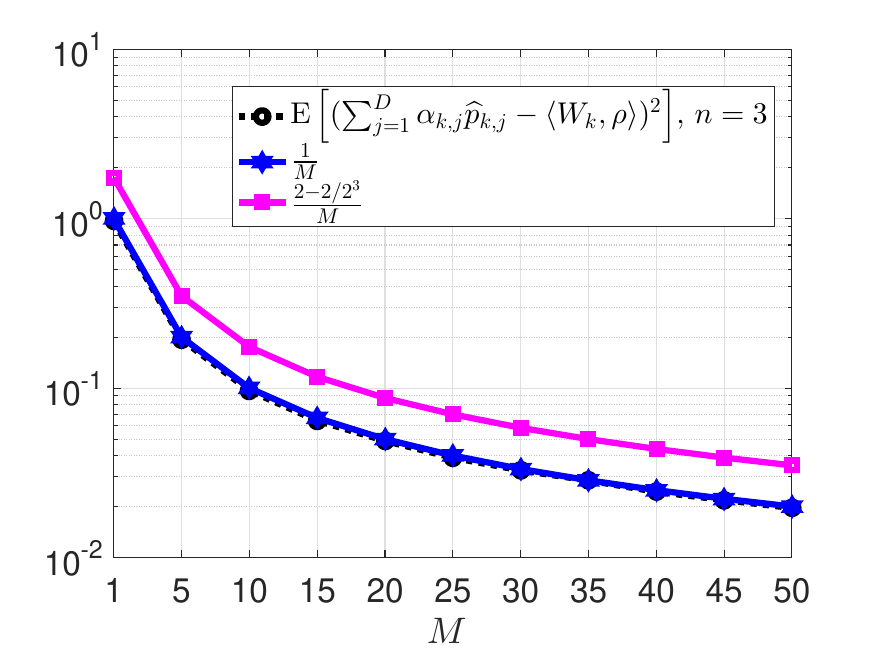}
\end{minipage}
\label{n = 3 error bound}
}
\subfigure[]{
\begin{minipage}[t]{0.4\textwidth}
\centering
\includegraphics[width=6.5cm]{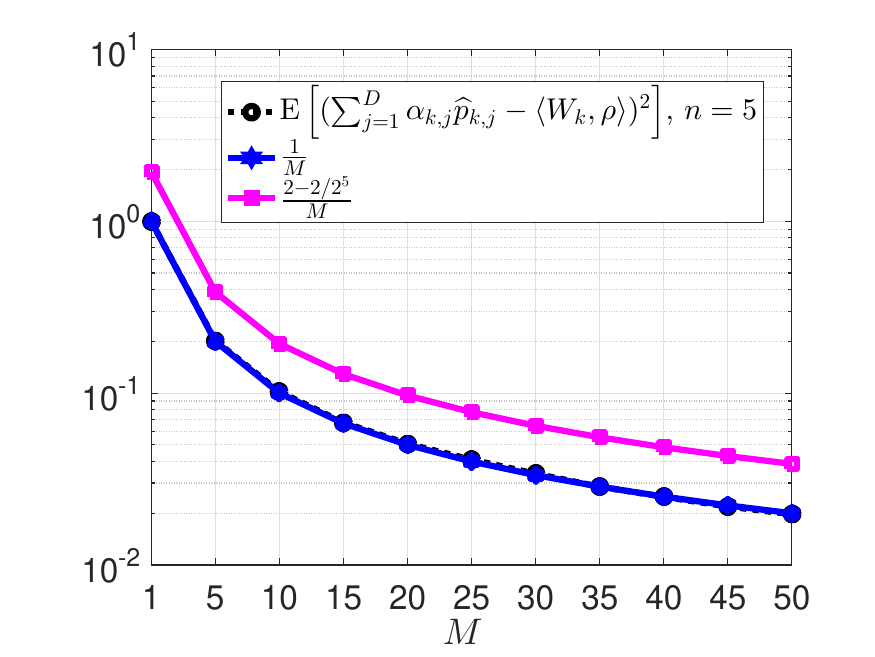}
\end{minipage}
\label{n = 5 error bound}
}
\subfigure[]{
\begin{minipage}[t]{0.4\textwidth}
\centering
\includegraphics[width=6.5cm]{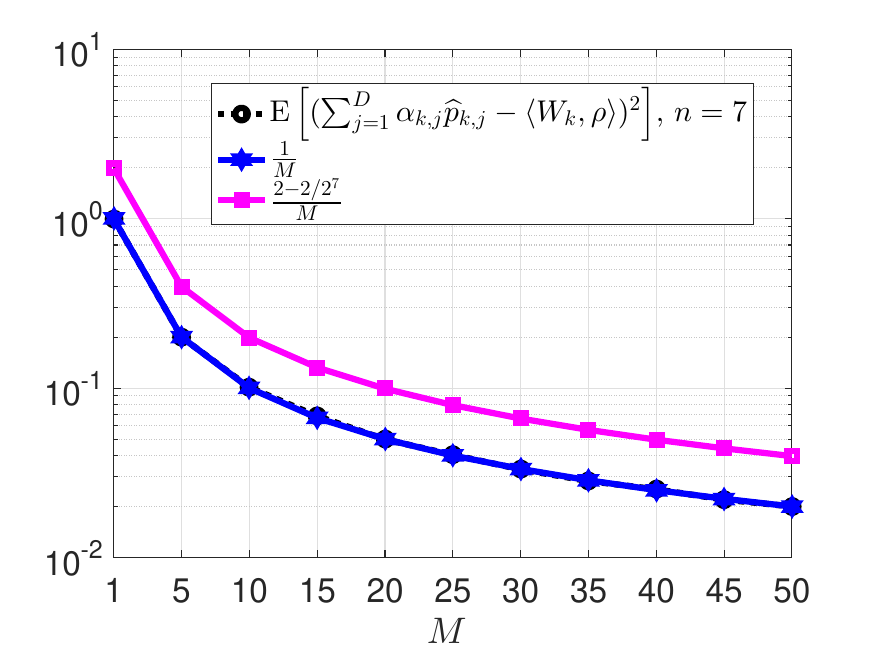}
\end{minipage}
\label{n = 7 error bound}
}
\caption{Numerical computation of $\E [(\sum_{j=1}^D \alpha_{k,j} \cdot \wh p_{k,j} - \<\mW_k, \vrho \> )^2 ]$ for different $n$ and $M$.}
\end{figure}

\section{Recovery Guarantee for Pauli Observable Measurements}
\label{Recovery guarantee for Pauli}

In this section, we will study the recovery of a ground truth density matrix $\vrho^\star$ from the above-mentioned Pauli observable measurements.

\subsection{Measurement setting and the matrix factorization approach}
\begin{definition}[Measurement setting]
    Suppose we draw $\mA_1,\ldots,\mA_K$ i.i.d.\ uniformly at random (with replacement) from the Pauli matrices $\{\mW_1,\ldots,\mW_{D^2}\}$ defined in \eqref{eq:Pauli-matrix}. We can generate $K$ Pauli observables for $\vrho^\star$ through the linear measurement operator $\calA: \C^{D\times D} \rightarrow \R^K$ as
\begin{eqnarray}
\label{The defi of Pauli measurement sec:3}
\vy = \calA(\vrho^\star) = \begin{bmatrix}
          y_1 \\
          \vdots \\
          y_K
        \end{bmatrix} = \begin{bmatrix}
          \<\mA_1, \vrho^\star \> \\
          \vdots \\
          \<\mA_K, \vrho^\star \>
        \end{bmatrix} = \begin{bmatrix}
          \sum_{j=1}^D \alpha_{1,j} \cdot  p_{1,j} \\
          \vdots \\
          \sum_{j=1}^D \alpha_{K,j} \cdot  p_{K,j}
        \end{bmatrix},
\end{eqnarray}
where $\{p_{k,1},\ldots,p_{k,D}\}$ are the Pauli basis measurements in \eqref{defi of Pauli basis measurements}.
Note that $K$ also denotes the number of POVMs. For each $D$-outcome POVM, suppose we repeat the measurement process $M$ times and take the average of the outcomes to generate empirical probabilities
\begin{align}
\wh p_{k,j} = \frac{f_{k,j}}{M}, \ k\in[K], \ j\in[D],
\label{eq:empirical-prob sec:3}\end{align}
where $f_{k,j}$ denotes the number of times the $j$-th output is observed when using the $k$-th POVM $M$ times. We further denote the empirical Pauli observable measurements as
\begin{eqnarray}
\label{The defi of empirical Pauli measurement sec:3}
\wh \vy =  \begin{bmatrix}
          \wh y_1 \\
          \vdots \\
          \wh y_K
        \end{bmatrix} =  \begin{bmatrix}
          \sum_{j=1}^D \alpha_{1,j} \cdot  \wh p_{1,j} \\
          \vdots \\
          \sum_{j=1}^D \alpha_{K,j} \cdot  \wh p_{K,j}
        \end{bmatrix}.
\end{eqnarray}
We denote by $\ve$ the error in the empirical Pauli observable measurements:
\begin{eqnarray}
\label{The measurement error in sec 3}
\ve =  \wh \vy - \vy \in \R^K.
\end{eqnarray}
\label{def:measurement-setting}
\end{definition}

Before introducing the loss function, it is helpful to incorporate a low-dimensional structure into the ground truth state $\vrho^\star$. The dimension of a quantum system grows exponentially with its components, such as the number of particles in the system. Consequently, the matrix size of $\vrho^\star$ becomes very large even for a moderately sized quantum system.
To overcome challenges related to the direct estimation of the density matrix, such as computational and storage costs, we can exploit the inherent structure of pure or nearly pure quantum states characterized by low entropy and represented as low-rank density matrices \cite{KuengACHA17,guctua2020fast,francca2021fast,voroninski2013quantum,haah2017sample}.
However, dealing with the rank constraint in the density matrix often requires an expensive singular value decomposition (SVD), leading to high computational complexity. To address this, we propose a Burer-Monteiro type decomposition \cite{burer2003nonlinear, burer2005local} to represent the density matrix as follows:
\begin{eqnarray}
\label{The low-rank density matrix}
\vrho = \mU \mU^{\rm H}, \ \text{where} \ \mU\in\C^{2^n\times r} \ \text{and} \ \| \mU\|_F = 1.
\end{eqnarray}
Based on this parameterization, we might consider a non-convex optimization problem as follows:
\begin{eqnarray}
\label{The loss function in terms of U only}
\minimize_{\mU \in \C^{2^n\times r}}  f(\mU) = \frac{D}{2K}\bigg\|\calA(\mU\mU^{\rm H}) -  \wh \vy \bigg\|_2^2,
\end{eqnarray}
where $\calA$ is the induced linear map as defined in \eqref{The defi of Pauli measurement sec:3}.

We note, however, that as defined, $f(\mU)$ is non-holomorphic in the complex variable domain, rendering the computation and analysis of gradients and Hessians cumbersome. For this reason, we consider instead the problem\footnote{Note that $\mU\mU^{\rm H}$ is consistently used as shorthand for $\mU {\mU^*}^\top$.}
\begin{eqnarray}
\label{The loss function in terms of U}
\minimize_{\mU,\mU^*\in\C^{D\times r}} f(\mU, \mU^*)   = \frac{D}{2K}\bigg\|\calA(\mU {\mU^*}^\top) -  \wh \vy \bigg\|_2^2,
\end{eqnarray}
where $\mU^*$, representing the complex conjugate of the matrix $\mU$, is considered a second variable in the optimization. The benefit of parameterizing the objective function by both $\mU$ and $\mU^*$ is that $f(\mU, \mU^*)$ remains holomorphic in $\mU$ for a fixed $\mU^*$, and vice versa, allowing the use of Wirtinger derivatives to compute gradients and Hessians. While one could alternatively parameterize $f(\mU)$ by $\Re{\mU}$ and $\Im{\mU}$, this would increase the complexity of the analysis. We refer readers to \cite[Chapter 3]{zhang2017matrix} for a comprehensive discussion on complex matrix analysis.

We also note that the solution set presented in \eqref{The loss function in terms of U} only comprises low-rank and PSD matrices, omitting the unit trace constraint. This omission arises because imposing the unit Frobenius norm does not notably diminish the magnitude of the recovery error.
We refer readers to the additional discussions after \Cref{upper bound recovery error of QST for Pauli measurements} and Corollary \ref{LocalConvergence_Conclusion_1}. Thus, we will exclusively concentrate on the set of low-rank and PSD matrices without the unit trace constraint.

Despite $\ell_2$ loss being a convex function, the bilinear form in \eqref{The low-rank density matrix} renders the objective function of \eqref{The loss function in terms of U} nonconvex. This indicates that \eqref{The loss function in terms of U} may contain many saddle points or spurious local minima at which algorithms like Wirtinger gradient descent could get stuck. Therefore, it is crucial to conduct a thorough analysis of the landscape of  \eqref{The loss function in terms of U}
to better understand the distribution and characteristics of saddle points and local minima. Specifically, through the following analysis under the restricted isometry property (RIP) for Pauli observable measurements, we divide the objective function landscape into two regions. In one region, any critical points must be strict saddle points, at which algorithms like Wirtinger gradient descent will not get stuck. This implies that any local minima in \eqref{The loss function in terms of U} must reside in the other region. Then, by delving further into the information gleaned from the local minima, we explore the relationship between the total number of state copies $KM$ and recovery error.

\paragraph{Restricted isometry property with Pauli observable measurements}
Our goal is to recover the low-rank density matrix $\vrho^\star$ from the underdetermined set of empirical Pauli observable measurements $\wh \vy$, where the number of coefficients $K$ is much smaller than the total number of entries in $\vrho^\star$, i.e., $K\ll D^2$. The key property necessary to achieve this is the restricted isometry property (RIP). The following results establish the RIP for $\calA$ over low-rank matrices.
\begin{theorem} (\cite[Theorem 2.1]{liu2011universal})
\label{thm:RIP-Pauli}
Let the linear map $\calA: \C^{D\times D} \rightarrow \R^K$ be defined in \eqref{The defi of Pauli measurement sec:3} and  $0\le \delta_r<1$. When the number of coefficients satisfies
\begin{eqnarray}
\label{The requirement in K of Pauli RIP}
K \geq C\cdot \frac{1}{\delta_r^2} Dr(\log D)^6
\end{eqnarray}
for some constant $C$, then, with overwhelming probability $1 - e^{-C}$ over matrices $\mA_1,\ldots,\mA_K$ selected i.i.d.\ uniformly from the set of Pauli matrices $\{\mW_1,\ldots,\mW_{D^2}\}$,  $\calA$ satisfies the $(r,\delta_r)$-RIP. That is, for all rank-$r$ $\vrho$, we have
\begin{eqnarray}
    \label{The definition of Unfolding}
(1-\delta_r)\norm{\vrho}{F}^2 \le \frac{D}{K}\|\calA(\vrho)\|_2^2 \le (1+\delta_r)\norm{\vrho}{F}^2.
\end{eqnarray}
\end{theorem}
In accordance with equation \eqref{The definition of Unfolding}, when $\vrho$ has low rank, the RIP guarantees that the energy $\frac{D}{K}\|\calA(\vrho)\|_2^2$ is close to $\|\vrho\|_F^2$. The RIP can then be used to guarantee the recovery of $\vrho$ using only $O(Dr(\log D)^6)$ Pauli observables. For example, for any two distinct low-rank Hermitian matrices $\vrho_1, \vrho_2$ with rank $r$, the $(2r,\delta_{2r})$-RIP guarantees distinct measurements since
\begin{eqnarray}
    \label{RIP guarantee for any matrices}
    \frac{D}{K}\|\calA(\vrho_1)- \calA(\vrho_2)\|_2^2 = \frac{D}{K}\|\calA(\vrho_1 - \vrho_2)\|_2^2 \ge (1- \delta_{2 r}) \|\vrho_1 - \vrho_2\|_F^2.
\end{eqnarray}

\subsection{Global landscape analysis}
Before analyzing the recovery error, we need to characterize the locations of the critical points of $f(\mU, \mU^*)$. It is also necessary to characterize the second derivative behavior at these critical points. We commence by introducing the concepts of critical points, strict saddles, and the strict saddle property. Unlike in real matrix analysis, our consideration extends beyond $\mU$ to also encompass $\mU^*$ in the complex variable function, as $\mU$ and $\mU^*$ are regarded as two independent variables. For a comprehensive understanding of complex variable functions, please refer to \cite[Chapters 3-4]{zhang2017matrix}.

\begin{definition}(Critical points)
\label{defi of critical points}
We say $(\mU,\mU^*)$ is a critical point if the Wirtinger gradient at $(\mU,\mU^*)$ vanishes, i.e.,
\begin{eqnarray}
    \label{Gradient form sec:3}
    \nabla f(\mU,\mU^*)=\begin{bmatrix}
    \nabla_{\mU} f(\mU,\mU^*) \\
          \nabla_{\mU^*} f(\mU,\mU^*)
        \end{bmatrix}= {\bm 0}.
\end{eqnarray}
\end{definition}
Note that $(\nabla_{\mU} f(\mU,\mU^*))^* = \nabla_{\mU^*} f(\mU,\mU^*)$ always holds for a real-valued function $f(\mU,\mU^*)$.

\begin{definition}(Strict saddles)
\label{Strict saddles and second-order critical points}
A critical point $(\mU,\mU^*)$ is a strict saddle if the Wirtinger Hessian matrix evaluated at this point has a strictly negative eigenvalue, i.e.,
\begin{eqnarray}
    \label{Hessian form sec:3}
    \lambda_{\min}(\nabla^2 f(\mU,\mU^*)) = \lambda_{\min}\bigg(\begin{bmatrix}
    \frac{\partial^2 f(\mU,\mU^*)}{\partial{\text{vec}(\mU^*)}\partial{\text{vec}(\mU)^\top}} \frac{\partial^2 f(\mU,\mU^*)}{\partial{\text{vec}(\mU^*)}\partial{\text{vec}(\mU^*)^\top}} \\
          \frac{\partial^2 f(\mU,\mU^*)}{\partial{\text{vec}(\mU)}\partial{\text{vec}(\mU)^\top}} \frac{\partial^2 f(\mU,\mU^*)}{\partial{\text{vec}(\mU)}\partial{\text{vec}(\mU^*)^\top}}
        \end{bmatrix}\bigg) <0.
\end{eqnarray}
All other critical points that satisfy the second-order optimality condition, i.e., $\nabla^2 f(\mU,\mU^*) \succeq 0$, are called second-order critical points.
\end{definition}
Note that without distinguishing between local maxima and saddle points, we refer to a critical point as a strict saddle point only if its Wirtinger Hessian matrix has at least one strictly negative eigenvalue.
Next, following the analysis in Appendix \ref{Noisy Matrix Sensing}, we formally establish the following theorem.

\begin{theorem}
\label{Strict saddle points theorem of f}
Suppose that $\vrho^\star\in\C^{D\times D}$ is a target density matrix of rank $r$  and linear map $\calA$ satisfies the $(2r,\delta_{2r})$-RIP with $\delta_{2r} \le 0.09$. Then any second-order critical point $(\mU,\mU^*)$ in the optimization problem \eqref{The loss function in terms of U} satisfies
\begin{eqnarray}
\label{condition of all second-order critical points in the main paper}
\norm{\mU\mU^{\rm H} - \vrho^\star}{F} \leq D\sqrt{r} h(\delta_{2r}) \frac{\|\calA^*(\ve)\|}{K} ,
\end{eqnarray}
where $\ve$ is defined in \eqref{The measurement error in sec 3}, $\calA^*$ is the adjoint operator of $\calA$ with $\calA^*(\ve)=\sum_{k=1}^K e_k\mA_k$, and $h(\delta_{2r})$ is defined as
\begin{eqnarray}
\label{h-delta-function}
h(\delta_{2r}) = \frac{\sqrt{2}+\sqrt{82.55- 762.54\delta_{2r}- 843.09\delta_{2r}^2}}{1.5-15.7\delta_{2r}}.
\end{eqnarray}
\end{theorem}

In words, \Cref{Strict saddle points theorem of f} implies that any critical point $(\mU,\mU^*)$ that is not close to the target solution, i.e., for which $
\norm{\mU\mU^{\rm H} - \vrho^\star}{F} \geq D\sqrt{r} h(\delta_{2r}) \frac{\|\calA^*(\ve)\|}{K}$, must be a strict saddle point. As various iterative algorithms can escape strict saddles \cite{lee2019first}, this property ensures the convergence to a solution with the recovery guarantee in \eqref{condition of all second-order critical points in the main paper}. The term $h(\delta_{2r})$ in the upper bound, as illustrated in Figure~\ref{Plot of h_delta}, is monotonically increasing with respect to $\delta_{2r}$, indicating that a smaller $\delta_{2r}$ could give a relatively better recovery as per \eqref{condition of all second-order critical points in the main paper}. The term $\|\calA^*(\ve)\|/K$ in \eqref{condition of all second-order critical points in the main paper} can be upper bounded by the following result.

\begin{figure}[thbp]
\centering
\includegraphics[width=8cm, keepaspectratio]%
{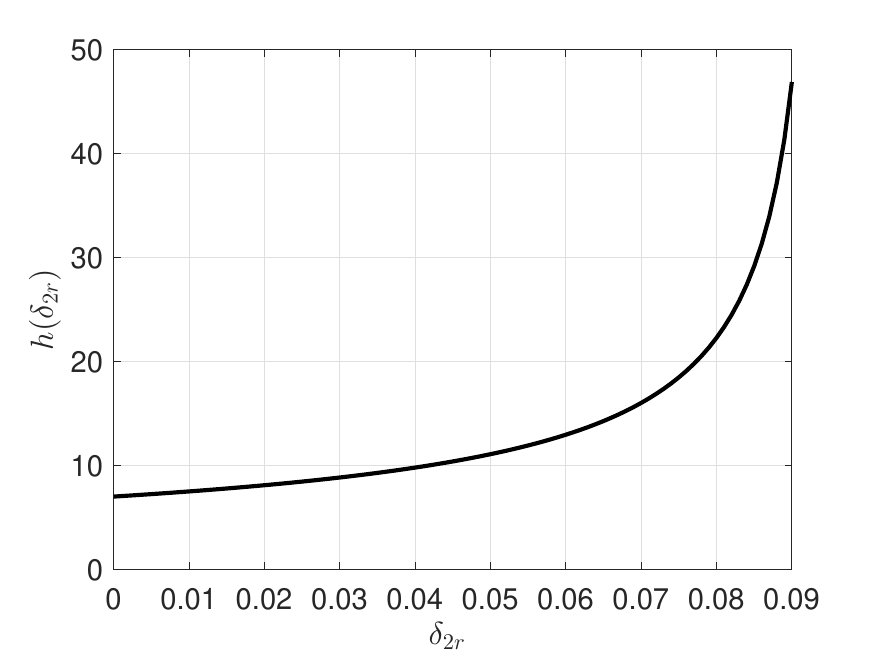}
\caption{Plot of $h(\delta_{2r})$ in \eqref{h-delta-function}.}
\label{Plot of h_delta}
\end{figure}

\begin{lemma}
\label{upper bound A* e}
Assuming $M \leq O(K/n)$, then with probability at least $1 - e^{-\Omega(n /(1 + \sqrt{Mn/K}) )}$, $\calA$ in \Cref{def:measurement-setting} satisfies
\begin{eqnarray}
\label{matrix Bernstein's inequality conclusion main paper}
\frac{\|\calA^*(\ve)\|}{K}  \lesssim \sqrt{\frac{\log D}{KM}}.
\end{eqnarray}
\end{lemma}
The proof is provided in Appendix~\ref{error bound of A*(e)}. \Cref{upper bound A* e} shows that the
term $\|\calA^*(\ve)\|/K$ caused by the statistical error in the empirical Pauli observable measurements depends only on the overall experimental resources $N = KM$, if ignoring the difference in the failure probability. This suggests that the recovery error remains comparable across various selections of $K$ and
$M$, provided that $N=KM$ remains constant. However, we note that the probability of failure may increase significantly when $K$ is excessively small. The assumption $M \leq O(K/n)$ stems from a technical difficulty in applying the matrix Bernstein’s inequality, but we believe this assumption can be removed by more sophisticated analysis.  By incorporating this result into \Cref{Strict saddle points theorem of f}, we can then obtain a recovery guarantee in terms of $K$ and $M$.

\begin{theorem}
\label{upper bound recovery error of QST for Pauli measurements}
Under the Pauli observable setting in \Cref{def:measurement-setting}, assume that $K \geq C\cdot \frac{1}{\delta_{2r}^2} Dr(\log D)^6$ for a positive constant $C$ such that $\calA$ satisfies $(2r,\delta_{2r})$-RIP with constant $\delta_{2r} \le 0.09$, as guaranteed by \Cref{thm:RIP-Pauli} with high probability. Supposing $M \leq O(K/n)$, then with probability at least $1 - e^{-C/2} - e^{-\Omega(n /(1 + \sqrt{Mn/K}) )} $, any second-order critical point $(\mU,\mU^*)$ in the optimization problem \eqref{The loss function in terms of U} satisfies
\begin{eqnarray}
\label{upper bound recovery error in the main paper}
\|\mU\mU^{\rm H} - \vrho^\star\|_F \lesssim h(\delta_{2r})\cdot \sqrt{\frac{D^2\log Dr}{KM}},
\end{eqnarray}
where $h(\delta_{2r})$ is defined in \eqref{h-delta-function}.
\end{theorem}

\paragraph{Quality versus quantity of Pauli observable measurements} When ignoring the term $h(\delta_{2r})$, the upper bound for the recovery error $\|\mU\mU^{\rm H} - \vrho^\star \|_F$ in \eqref{upper bound recovery error in the main paper} decays with the overall experimental resources $N = KM$. Specifically, a favorable scaling of $1/KM$ for the mean-squared error (i.e., $\|\mU\mU^{\rm H} - \vrho^\star \|_F^2$) is maintained for all sufficiently large $K$ such that the RIP is satisfied.  On the other hand, the constant $\delta_{2r}$ decays with $K$ and $h(\delta_{2r})$ increases monotonically, albeit slowly, with $\delta_{2r}$ as depicted in \Cref{Plot of h_delta}. From this perspective, for fixed total experimental resources $N$, using a larger value of $K$ generally results in a smaller recovery error bound. In particular, the extreme case with a single-sample measurement (i.e., $M = 1$ and $K = N$), where each empirical Pauli observable measurement is extremely noisy, would give the best bound. This explains the successful application of single-sample measurements in recent work \cite{wang2020scalable}. We will provide experiments in \Cref{Numerical Experiments} to demonstrate the recovery performance for the different choices of $M$ and $K$.

For any $\epsilon>0$, \eqref{upper bound recovery error in the main paper} implies that $\|\mU\mU^{\rm H} - \vrho^\star \|_F\leq \epsilon$
when the number of state copies $KM \gtrsim  D^2(\log D)r/\epsilon^2$.
By using the inequality $\|\mU\mU^{\rm H} - \vrho^\star\|_1 \leq \sqrt{2r}\norm{\mU\mU^{\rm H} - \vrho^\star}{F}$ since $\rank(\mU\mU^{\rm H} - \vrho^\star) \le 2r$, we can ensure $\epsilon$-accurate recovery in the trace distance with $\Omega (D^2(\log D)r^2/\epsilon^2)$ state copies. When omitting the $\log D$ term, which might have been introduced for technical reasons, this sampling complexity is optimal for empirical Pauli observable measurements according to \cite{flammia2012quantum}. Finally, we note that while the solution $\mU\mU^{\rm H}$ of \eqref{The loss function in terms of U} may be non-physical as it may not have trace 1, we can easily obtain a physical state without compromising the recovery guarantee in \Cref{upper bound recovery error in the main paper}. Specifically, denoting by $\setS_+:=\{\vrho \in \C^{2^n\times 2^n}: \vrho \succeq \vzero, \trace(\vrho) = 1\}$ the set of physical states and $P_{\setS_+}$ the projection onto the set $\setS_+$ which can be efficiently computed by projecting the eigenvalues onto a simplex \cite{condat2016fast}, we have
\begin{eqnarray}
\label{upper bound recovery error in the main paper physical}
\|\calP_{\setS_+}(\mU\mU^{\rm H}) - \vrho^\star\|_F \leq \|\mU\mU^{\rm H} - \vrho^\star\|_F,
\end{eqnarray}
where the second inequality follows from the nonexpansiveness property of the convex set. This implies that the projection step ensures the state becomes physically valid while preserving or even improving the recovery guarantee.

\section{Wirtinger Gradient Descent with Linear Convergence}
\label{Convergence analysis of GD}

\Cref{upper bound recovery error of QST for Pauli measurements} provides an upper bound on the recovery error for any second-order critical point of the least squares objective with empirical Pauli observable measurements. This result offers favorable guarantees for many iterative algorithms, including the computationally efficient gradient descent, which can almost surely converge to a second-order critical point \cite{lee2019first}. However, this does not imply that gradient descent can efficiently find a second-order critical point. Indeed, gradient descent can be significantly slowed down by saddle points and may take exponential time to escape \cite{du2017gradient}. In this section, we study the optimization landscape near the target solution and show that the problem has a basin of attraction within which Wirtinger gradient descent has a fast convergence, with recovery error similar to \eqref{condition of all second-order critical points in the main paper}.

In particular, with an appropriate initialization $\mU_0$, we use the following Wirtinger gradient descent (GD):
\begin{eqnarray}
    \label{UpdatingOfU}
    {\mU}_t = \mU_{t-1} -\mu \nabla_{\mU^*} f(\mU_{t-1},\mU^*_{t-1}),
\end{eqnarray}
where $\mu$ is the step size. The complex conjugate part can also be updated as ${\mU}_t^* = \mU_{t-1}^* -\mu \nabla_{\mU} f(\mU_{t-1},\mU^*_{t-1})$, which is exactly the same as taking the complex conjugate of \eqref{UpdatingOfU} for the real-valued $f(\mU,\mU^*)$. Thus, one only needs to perform the gradient descent and study the performance for $\mU_t$.  To analyze the convergence, we need an appropriate metric to capture the distance between the factors $\mU_t$ and the ground-truth factor $\mU^\star$ that satisfies $\vrho^\star = \mU^\star\mU^{\star\H}$. Noting that for any unitary matrix $\mR \in \calU_r:=\{\mR\in \C^{r\times r}, \mR^\H\mR = \mId\}$, $\mU^\star \mR$ is also a factor of $\vrho^\star$ since $\mU^\star\mR\mR^{\rm H}\mU^{\star\rm H} = \vrho^\star$, we define the distance between any factor $\mU$ and $\mU^\star$ as \cite{Tu16}
\begin{eqnarray}
    \label{distance between two matrices}
    \dist(\mU,\mU^\star) = \|\mU-\mU^\star\mR \|_F, \ \text{where} \ \mR = \argmin_{\mR'\in\calU_r} \|\mU-\mU^\star\mR' \|_F.
\end{eqnarray}
Based on the distance defined above, we can analyze the convergence of the Wirtinger GD by
\begin{eqnarray}
    \label{expansion of distance: local convergence}
    \dist^2(\mU_t,\mU^\star) &\!\!\!\! = \!\!\!\!& \|\mU_{t}-\mU^\star\mR_{t} \|_F^2\nonumber\\
    &\!\!\!\!\leq\!\!\!\!& \|\mU_{t-1} -\mu \nabla_{\mU^*} f(\mU_{t-1}, \mU_{t-1}^*) - \mU^\star\mR_{t-1}\|_F^2\nonumber\\
    &\!\!\!\!=\!\!\!\!&\|\mU_{t-1} - \mU^\star\mR_{t-1}\|_F^2 + \mu^2 \|\nabla_{\mU^*} f(\mU_{t-1},\mU_{t-1}^*)\|_F^2\nonumber\\
    &\!\!\!\!\!\!\!\!&- 2\mu\Re{\<\mU_{t-1} - \mU^\star\mR_{t-1}, \nabla_{\mU^*} f(\mU_{t-1},\mU_{t-1}^*) \>}.
\end{eqnarray}
To enable a sufficient decrease of the distance, the term $\Re{\<\mU_{t-1} - \mU^\star\mR_{t-1}, \nabla_{\mU^*} f(\mU_{t-1},\mU_{t-1}^*) \>}$ needs to be sufficiently large, i.e., the search direction $-\nabla_{\mU^*} f(\mU_{t-1},\mU_{t-1}^*)$ needs to point towards to the target factor $\mU^\star \mR_{t-1}$, which can be guaranteed by strong convexity for convex functions. For real-valued matrix factorization problems, a regularity condition has been established to achieve that goal \cite[Lemma 5.7]{Tu16}. Here, we extend the analysis to the complex domain.

\begin{lemma} (Regularity Condition)
\label{lem:RegularityCondition_f}
Under the Pauli observable setting in \Cref{def:measurement-setting}, assume that $K \geq C\cdot \frac{1}{\delta_{3r}^2} Dr(\log D)^6$ for a positive constant $C$ such that $\calA$ satisfies $(3r,\delta_{3r})$-RIP with constant $\delta_{3r} \le 1/40$, as guaranteed by \Cref{thm:RIP-Pauli}. Then, with probability at least $1 - e^{-C/3}$, any $\mU$ that is close to $\mU^\star$ in the sense $
\dist(\mU,\mU^\star)\leq\frac{\sigma_r(\mU^\star)}{4}$ satisfies
\begin{eqnarray}
\label{RegularityCondition_Result_f}
&\!\!\!\!\!\!\!\!&\Re{\<\nabla_{\mU^*} f(\mU,\mU^*), \mU-\mU^\star\mR \>}\nonumber\\
&\!\!\!\!\geq\!\!\!\!&\frac{\sigma_r^2(\mU^\star)}{4}\|\mU-\mU^\star\mR \|_F^2+ \frac{2(1-4\delta_{3r})}{25(1+2\delta_{3r})^2}\|\nabla_{\mU^*} f(\mU,\mU^*)\|_F^2-\frac{11D^2r}{K^2}\|\calA^*(\ve)\|^2.
\end{eqnarray}
\end{lemma}

The proof is given in {Appendix} \ref{Proof_LocalConvergence_Conclusion 1}. Towards interpreting the right hand side (RHS) of \eqref{RegularityCondition_Result_f}, first recall  \Cref{upper bound A* e} and the subsequent discussion that the term $\|\calA^*(\ve)\|/K$ decays with the overall experimental resources on the order of $1/\sqrt{KM}$. On the other hand,
a larger $K$ can result in a smaller $\delta_{3r}$, and hence a larger $\frac{2(1-4\delta_{3r})}{25(1+2\delta_{3r})^2}$. This aligns with the observation in \Cref{upper bound recovery error of QST for Pauli measurements} that performance roughly depends on the entire experimental resources $KM$, but increasing $K$ can enhance performance. Note that the first two terms in the RHS of \eqref{RegularityCondition_Result_f} decay as $\mU$ approaches $\mU^\star$, while the third term is constant for any $\mU$. Thus, roughly speaking, the RHS is positive when $\mU$ is relatively far from $\mU^\star$, ensuring a sufficient decrease of one step of the Wirtinger GD update. However, it becomes negative when $\mU$ is close to $\mU^\star$, indicating that no further decrease can be ensured. Formally, plugging \Cref{lem:RegularityCondition_f} into \eqref{expansion of distance: local convergence} gives
\begin{eqnarray}
\label{LocalConvergence_Conclusion_3}
    \|\mU_{t}-\mU^\star\mR_{t} \|_F^2
    &\!\!\!\!\leq\!\!\!\!&\bigg(1-\frac{\sigma_r^2(\mU^\star)}{2}\mu\bigg)\|\mU_{t-1}-\mU^\star\mR_{t-1} \|_F^2+\frac{22D^2r \mu}{K^2}\|\calA^*(\ve)\|^2\nonumber\\
    &\!\!\!\!\leq\!\!\!\!& \bigg(1-\frac{\sigma_r^2(\mU^\star)}{2}\mu\bigg)^{t}\|\mU_{0}-\mU^\star\mR_{0} \|_F^2 + O\bigg(\frac{D^2(\log D)r}{\sigma_r^2(\mU^\star)KM}\bigg),
\end{eqnarray}
where the first inequality assumes $\mu\leq\frac{4(1-4\delta_{3r})}{25(1+2\delta_{3r})^2}$ and the second inequality uses \Cref{upper bound A* e}. We summarize this local convergence property of Wirtinger GD in the following corollary.
\begin{Corollary} (Local Convergence)
\label{LocalConvergence_Conclusion_1}
Under the same setup as in \Cref{lem:RegularityCondition_f}, suppose the Wirtinger GD \eqref{UpdatingOfU} starts with an initialization $\mU_0$ that satisfies
\begin{eqnarray}
\label{Initialization local convergence}
\|\mU_0-\mU^\star\mR_0 \|_F\leq\frac{\sigma_r(\mU^\star)}{4}
\end{eqnarray}
and uses step size $\mu\leq\frac{4(1-4\delta_{3r})}{25(1+2\delta_{3r})^2}$. Then with probability at least $1 - e^{-C/3} - e^{-\Omega(n /(1 + \sqrt{Mn/K}) )} $, we have
\begin{eqnarray}    \label{LocalConvergence_Conclusion_2}
    \|\mU_{t}-\mU^\star\mR_{t} \|_F^2 \leq \bigg(1-\frac{\sigma_r^2(\mU^\star)}{2}\mu\bigg)^{t}\|\mU_{0}-\mU^\star\mR_{0} \|_F^2 + O\bigg(\frac{D^2(\log D)r}{\sigma_r^2(\mU^\star)KM}\bigg).
\end{eqnarray}
\end{Corollary}
As above discussed, the first term in the RHS of \eqref{LocalConvergence_Conclusion_2} indicates that Wirtinger GD achieves a linear convergence rate in the first phase (when $\mU_t$ has a large distance to the target), which is similar to the noiseless real-valued cases~\cite{Tu16}. The second term in the RHS of \eqref{LocalConvergence_Conclusion_2}, which dominates after a certain number of iterations, is consistent with the recovery guarantee in \Cref{upper bound recovery error of QST for Pauli measurements} except for the additional term $\sigma_r^2(\mU^\star)$ that only depends on the target state. Notice that based on the initialization condition $\|\mU_0-\mU^\star\mR_0 \|_F\leq\frac{\sigma_r(\mU^\star)}{4}$, we can obtain $\|\mU_{t}\|\leq \|\mU_{t} - \mU^\star\mR_{t} \| + \|\mU^\star\mR_{t}\|\leq \frac{5}{4}$ and further derive
\begin{eqnarray}
    \label{relationship between nonphysical and physical}
     \|\vrho_t - \vrho^\star \|_F^2
    &\!\!\!\! = \!\!\!\!& \|\mU_{t}\mU_{t}^{\rm H}- \mU_{t}\mR_{t}^{\rm H}{\mU^\star}^{\rm H} +\mU_{t}\mR_{t}^{\rm H}{\mU^\star}^{\rm H} -\mU^\star{\mU^\star}^{\rm H} \|_F^2 \nonumber\\
      &\!\!\!\!\leq\!\!\!\!& (2\|\mU_{t}\| + 2\|\mU^\star\|)\|\mU_{t}-\mU^\star\mR_{t} \|_F^2 \leq O\bigg(\frac{D^2(\log D)r}{\sigma_r^2(\mU^\star)KM}\bigg).
\end{eqnarray}
Combing \eqref{upper bound recovery error in the main paper physical} and \eqref{relationship between nonphysical and physical}, we can conclude that the recovery error bound mentioned above is also applicable to the physical state.

\paragraph{Spectral initialization} To provide a good initialization that satisfies \eqref{Initialization local convergence}, we utilize the spectral initialization approach that has been widely used in the literature \cite{candes2015phase,Ma21TSP,tong2021accelerating}. Specifically,
\begin{eqnarray}
    \label{Spectral Initialization matrix}
 \mU_0 = \begin{bmatrix}(\lambda_1)_+^{1/2} \vc_1 & \cdots & (\lambda_r)_+^{1/2} \vc_r\end{bmatrix}, \quad \mC \mLambda \mC^{\rm H} = \frac{D}{K}\sum_{k=1}^K \wh y_k\mA_k,
\end{eqnarray}
where $(a)_+ = \max(a,0)$, $\vc_i$ denotes the $i$-th column of $\mC$, and $\mC\mLambda\mC^\H$ is the eigenvalue decomposition of $\frac{D}{K}\sum_{k=1}^K \wh y_k\mA_k$ with nonincreasing eigenvalues $\lambda_1\ge \cdots \ge \lambda_D$. The following result guarantees the performance of this initialization.

\begin{theorem} (Spectral Initialization)
\label{The spectral initialization in main paper}
Under the same setup as in \Cref{lem:RegularityCondition_f}, assuming $M\lesssim K/n$, then with probability at least $1 - e^{-C/3} - e^{-\Omega(n /(1 + \sqrt{Mn/K}) )}$, the spectral initialization \eqref{Spectral Initialization matrix} satisfies
\begin{eqnarray}
    \label{The final conclusion of intial conclusion analysis main paper}
    \|\mU_0-\mU^\star\mR_0 \|_F^2 \leq \frac{\bigg(4\delta_{3r}\|\vrho^\star\|_F + O\bigg(\sqrt{\frac{D^2(\log D)r}{KM}}\bigg)\bigg)^2}{2(\sqrt{2}-1)\sigma_r^2(\mU^\star)}.
\end{eqnarray}
\end{theorem}
The proof is provided in {Appendix} \ref{Proof_Upper bound of noisy spectral initialization}.
The RHS of \eqref{The final conclusion of intial conclusion analysis main paper} decays with $K$ and $M$, which can be chosen relatively large to ensure the spectral initialization satisfies the requirement in \eqref{Initialization local convergence}. We again observe a similar phenomenon in how $K$ and $M$ affect the RHS of \eqref{The final conclusion of intial conclusion analysis main paper}: the first term $\delta_{3r}$ decays with $K$, while the second term $O(\sqrt{\frac{D^2(\log D)r}{KM}})$ only decreases with the total number of state copies.

\section{Numerical Experiments}
\label{Numerical Experiments}

In this section, we present numerical experiments on quantum state tomography, focusing on low-rank density matrices, to demonstrate the applicability of our theoretical results. To begin, we randomly generate $K$ Pauli matrices and a rank-$r$ density matrix $\vrho^\star = \mU^\star {\mU^\star}^{\rm H}\in\C^{2^n\times 2^n}$, where $\mU^\star = \frac{ \mA^\star + i \cdot \mB^\star}{\|\mA^\star + i \cdot \mB^\star\|_F}\in \C^{2^n\times r}$ with the entries of $\mA^\star$ and $\mB^\star$ being drawn i.i.d.\ from the standard normal distribution. For all the following experiments, we simulate each over 20 Monte Carlo independent trials and compute the average over these 20 trials.

\begin{figure}[!ht]
\centering
\subfigure[]{
\begin{minipage}[t]{0.45\textwidth}
\centering
\includegraphics[width=6.5cm]{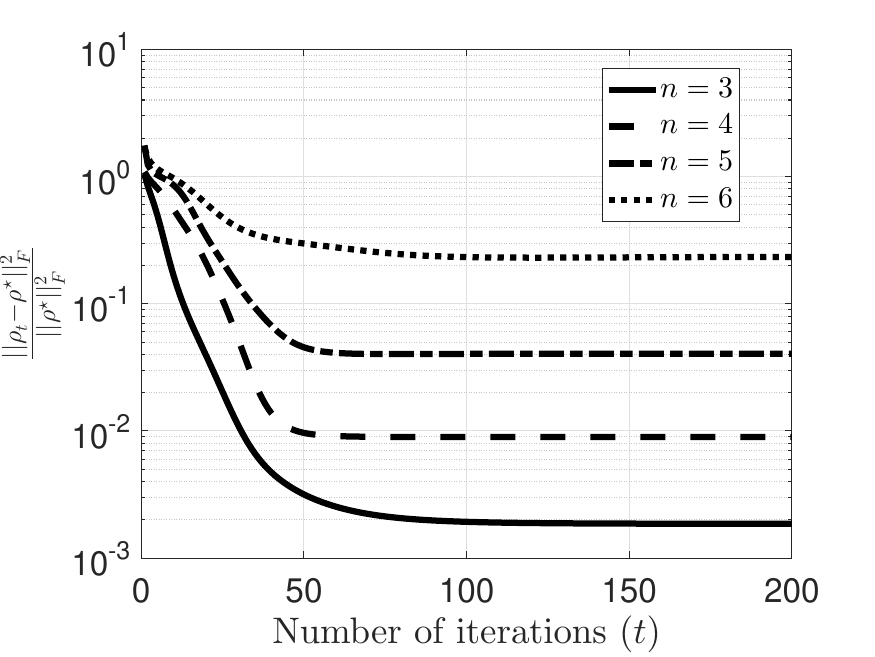}
\end{minipage}
\label{QST different n}
}
\subfigure[]{
\begin{minipage}[t]{0.45\textwidth}
\centering
\includegraphics[width=6.5cm]{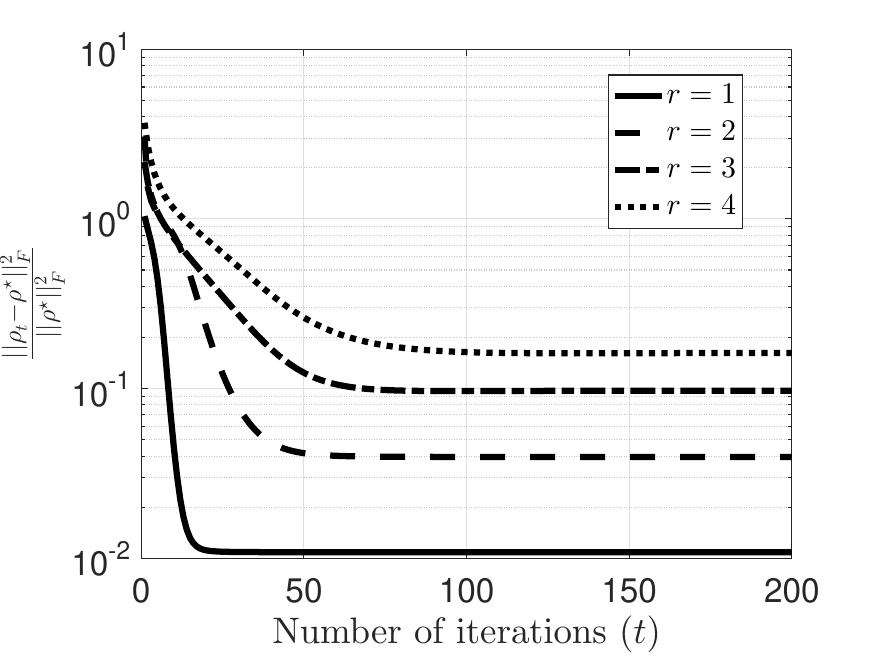}
\end{minipage}
\label{QST different r}
}
\caption{Performance comparison of Wirtinger GD, (a) for different $n$ with $r = 2$, $\mu = 0.3$, $K = 2000$ and $M = 100$, (b) for different $r$ with $n = 5$, $\mu = 0.3$, $K = 2000$ and $M = 100$.}
\end{figure}

\begin{figure}[!ht]
\centering
\subfigure[]{
\begin{minipage}[t]{0.45\textwidth}
\centering
\includegraphics[width=6.5cm]{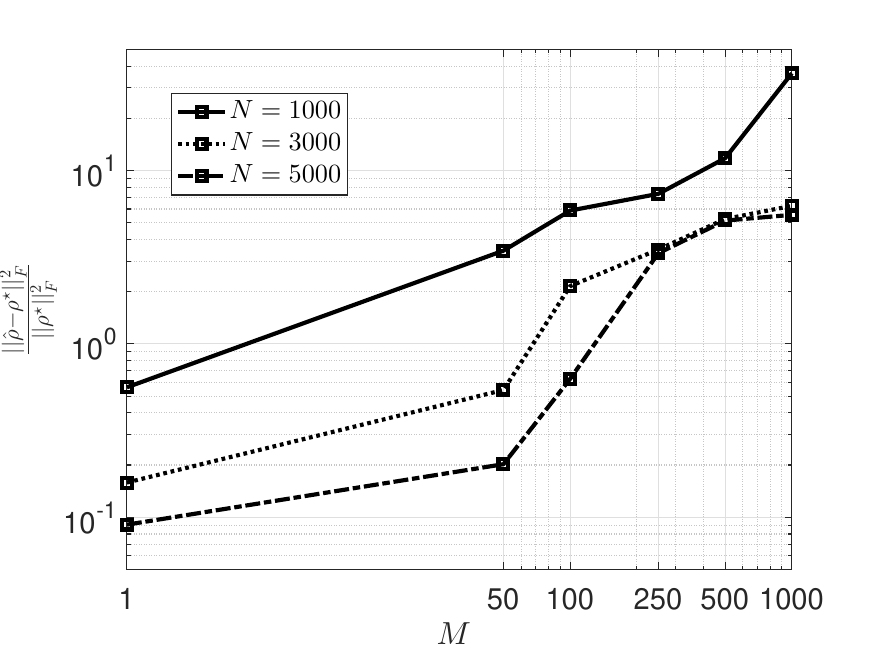}
\end{minipage}
\label{QST different n4r1}
}
\subfigure[]{
\begin{minipage}[t]{0.45\textwidth}
\centering
\includegraphics[width=6.5cm]{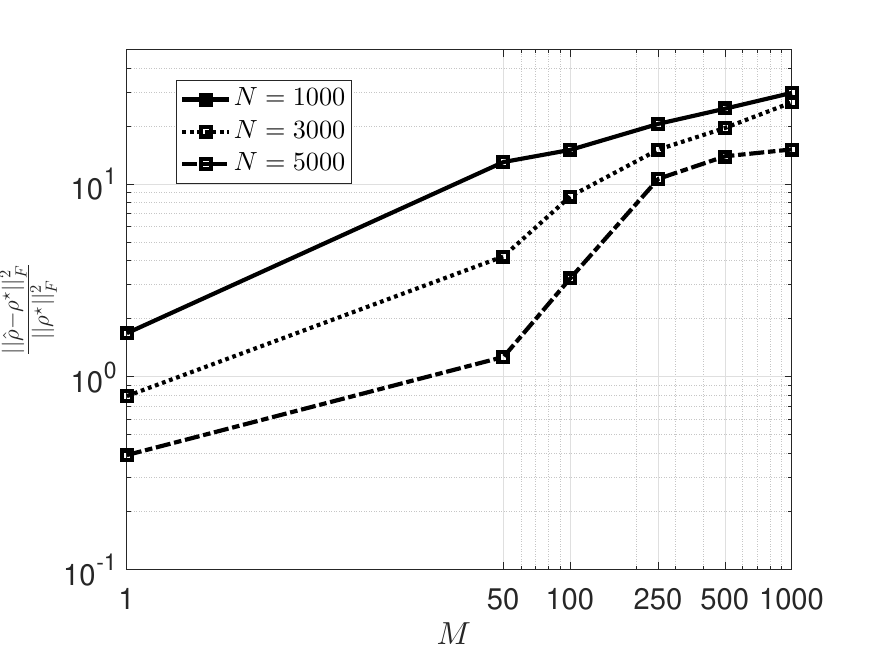}
\end{minipage}
\label{QST different n4r2}
}
\caption{Performance comparison of Wirtinger GD for different $N$ and $M$ with (a) $n = 4$, $r = 1$, (b) $n = 4$, $r = 2$.}
\end{figure}

\paragraph{Convergence of Wirtinger GD} In the first set of experiments, we evaluate the performance of  Wirtinger GD for various values of $n$ and $r$. Figures~\ref{QST different n} and \ref{QST different r} illustrate the convergence behavior of the algorithm in minimizing the loss function defined in \eqref{The loss function in terms of U}.
We notice that as the values of $n$ or $r$ increase, the convergence rate of the  Wirtinger GD  decreases while recovery errors increase, which is aligned with the findings in \Cref{LocalConvergence_Conclusion_1}.

\paragraph{Quality versus quantity of Pauli observable measurements}
In the second set of experiments, we investigate the trade-off between the number of POVMs and the number of repeated measurements by choosing different $M$ when the total number of state copies $N = KM$ is fixed. To ensure convergence, we set step sizes $\mu = 0.05$ for all the experiments. We plot the results for different choices of $r$ in Figures~\ref{QST different n4r1}-\ref{QST different n4r2}.
As anticipated, when $N$ is fixed, overall we observe a reduction in the recovery error as $M$ decreases (or $K$ increases). This is because a larger $K$  produces more measurements and reduces the RIP constant of the measurement operator in \eqref{The requirement in K of Pauli RIP}, overcoming the issue of higher statistical error in the measurements. This finding corroborates the results presented in \Cref{upper bound recovery error of QST for Pauli measurements} and the practical use of single-sample measurements \cite{wang2020scalable, huang2020predicting}.

\begin{figure}[!ht]
\centering
\includegraphics[width=6.5cm]{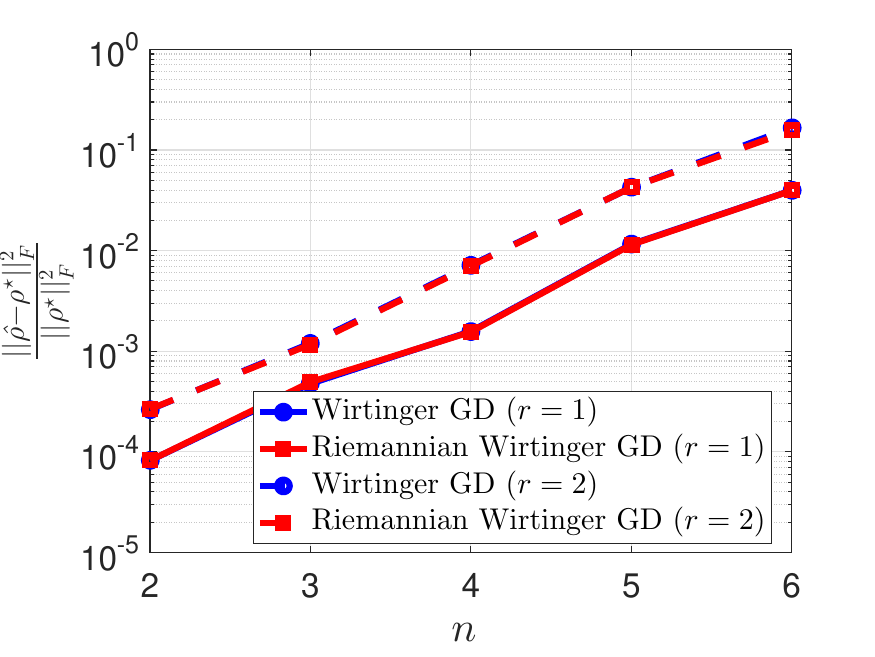}
\label{QST different algorithms local}
\caption{Performance comparison of different algorithms with $\mu = 0.1$, $K = 4000$ and $M = 50$.}
\end{figure}

\paragraph{Incorporating sphere constraint} In the previous analysis, we have omitted the unit trace constraint for the factor in \eqref{The low-rank density matrix} because imposing the unit Frobenius norm does not notably improve the performance. To support this argument, in the third set of experiments, we compare the approaches with and without this constraint. Specifically, we impose the unit trace constraint on the factors $\{\mU: \|\mU\|_F = 1  \}$ in \eqref{The loss function in terms of U} and then employ the Riemannian Wirtinger GD to solve the corresponding problem, with updates
\begin{eqnarray}
    \label{Riemannian GD}
    \wh {\mU}_t = \mU_{t-1} -\mu \calP_{T_{\mU} \text{Sp}}(\nabla_{\mU^*} f(\mU_{t-1},\mU^*_{t-1})) \ \ \text{and} \ \  {\mU}_t = \frac{\wh {\mU}_t}{\|\wh {\mU}_t\|_F},
\end{eqnarray}
where $\mu$ is the step size and $\calP_{T_{\mU} \text{Sp}}(\mV) = \mV - \<\mU, \mV \>\mU$ denotes the projection onto the tangent space $T_{\mU} \text{Sp} = \{\mQ\in\C^{D\times r}: \<\mU,\mQ\> = 0  \}$, giving a Riemannian gradient.
As shown in Figure~\ref{QST different algorithms local}, we observe that the two algorithms, Wirtinger GD and Riemannian Wirtinger GD, achieve on-par performance.

\begin{figure}[!ht]
\centering
\subfigure[]{
\begin{minipage}[t]{0.45\textwidth}
\centering
\includegraphics[width=6.5cm]{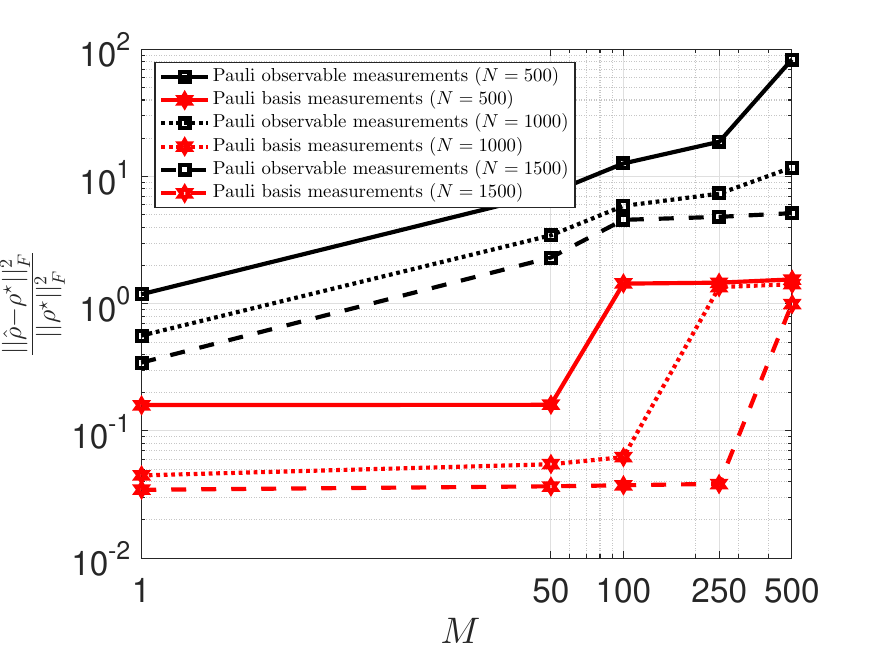}
\end{minipage}
\label{QST Local POVM different n4r1}
}
\subfigure[]{
\begin{minipage}[t]{0.45\textwidth}
\centering
\includegraphics[width=6.5cm]{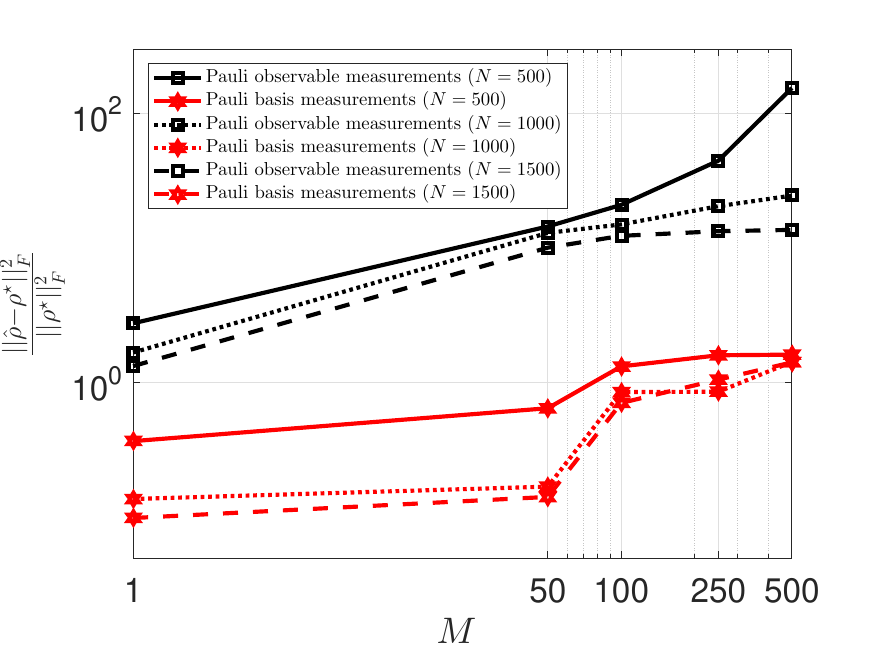}
\end{minipage}
\label{QST Local POVM different n4r2}
}
\caption{Performance comparison of the Wirtinger GD using empirical Pauli observable measurements and Pauli basis measurements with $\mu = 1$ for different $N$, $M$ and (a) $n = 4$, $r = 1$, (b) $n = 4$, $r = 2$.}
\end{figure}

\paragraph{Directly using Pauli basis measurements for QST}
Our work focused on the use of the expectation values of $n$-body Pauli observables in low-rank quantum state tomography to study the trade-off between the number of POVMs and the number of repeated measurements since the RIP with Pauli
observables facilitates the analysis. Note that the Pauli observables are actually measured using the local Pauli basis, one could also recover the target state directly from the empirical probabilities of the Pauli basis measurement (see \eqref{empirical Pauli basis measurements}).

Since Pauli basis measurements are also linear measurements of $\vrho^\star$, as described in \eqref{defi of Pauli basis measurements}, we can formulate a similar least squares minimization problem \eqref{The loss function in terms of U} and use the same Wirtinger GD to solve the problem. We conduct the numerical experiments under the same setting as the second set of experiments. By keeping the total number of measurements $N = KM$ constant, it becomes apparent that the Pauli basis measurements outperform Pauli observable measurements. One plausible explanation for this result is that the Pauli basis measurements involve $KD$ measurements, whereas the Pauli observable measurements involve only $K$  coefficients which are linear combinations of the Pauli basis measurements. On the other hand, with a fixed $N$, the recovery error associated with Pauli basis measurements also increases as $M$ grows. While we expect the analysis of the trade-off between the quality and quantity of measurements can be extended to Pauli basis measurements, one challenge is the lack of the RIP condition \cite{liu2011universal} for the Pauli basis measurements. We leave the investigation to future studies.

\section{Conclusion}
\label{conclusion}

In this paper, we investigate the trade-off between the number of POVMs and the number of repeated measurements that control the quality of measurements for quantum state tomography, focusing on low-rank states and Pauli observable measurements for theoretical analysis. By exploiting the matrix factorization approach, we first study the global landscape of the factorized problem and show that every second-order critical point is a good estimator of the target state, with better recovery achieved by using a greater number of measurement settings, despite the large statistical error in each measurement. This finding suggests the advantage of using single-sample measurements when the total number of state copies is kept constant. Additionally, we prove that Wirtinger gradient descent converges locally at a linear rate. Numerical simulation results support our theoretical findings.

\section*{Acknowledgment}
We acknowledge funding support from NSF Grants Nos. PHY-2112893, CCF-2106834, CCF-2241298 and ECCS-2409701 as well as the W. M. Keck Foundation. We thank the Ohio Supercomputer Center for providing the computational resources needed in carrying out this work.

\appendices

\section{Proof of \Cref{Strict saddle points theorem of f}}
\label{Noisy Matrix Sensing}
\begin{proof}
Recall the objective $f(\mU, \mU^*)  =  \frac{D}{2K}\big\|\calA(\mU\mU^{\rm H}) -  \wh \vy \big\|_2^2$ in \eqref{The loss function in terms of U}. We first derive its Wirtinger gradient as follows:
\begin{eqnarray}
    \label{Gradient form}
    \nabla f(\mU,\mU^*)=\begin{bmatrix}
    \nabla_{\mU} f(\mU,\mU^*) \\
          \nabla_{\mU^*} f(\mU,\mU^*)
        \end{bmatrix}=\frac{D}{K}\sum_{k=1}^K\begin{bmatrix}(\< \mA_k, \mU\mU^{\rm H}\> -  \wh y_k)\mA_k^* \mU^*\\ (\< \mA_k, \mU\mU^{\rm H}\> -  \wh y_k)\mA_k \mU  \end{bmatrix},
\end{eqnarray}
where we utilize the fact $\mA_k = \mA_k^{\rm H}$. It is noteworthy that within the framework of complex derivatives, the function $f(\mU)$ is written as $f(\mU, \mU^*)$ in the conjugate coordinates $(\mU, \mU^*)$. Therefore, we can independently compute the complex partial derivative $\nabla_{\mU} f(\mU,\mU^*)$ and the complex conjugate partial derivative $\nabla_{\mU^*} f(\mU,\mU^*)$, treating the complex variable $\mU$ and its complex conjugate $\mU^*$ as two independent variables.

Now we assume that $(\mU,\mU^*)$ is a critical point, i.e., $\nabla_{\mU} f(\mU,\mU^*) = {\bm 0}$ and $\nabla_{\mU^*} f(\mU,\mU^*) = {\bm 0}$. Let $\calP_{\mU} = \mU^{{\rm H}\dagger} \mU^{\rm H}$ be the orthogonal projection onto the range space of $\mU$.
Then, for any $\mZ\in\C^{D\times r}$, we have
\begin{eqnarray}
\label{The first order information}
   0 &&\!\!\!\!\!\!\!\!\!\!= \innerprod{\nabla_{\mU^*} f(\mU, \mU^*)}{\mZ} = \frac{D}{K}\innerprod{\calA^*\calA(\mU\mU^{\rm H} - \vrho^\star)}{\mZ\mU^{\rm H}} - \frac{D}{K}\innerprod{\calA^*(\ve)}{\mZ\mU^{\rm H}} \nonumber\\
  \Longrightarrow &&\!\!\!\!\!\!\!\!\!\!\innerprod{\mU\mU^{\rm H} - \vrho^\star}{\mZ\mU^{\rm H}} = \innerprod{\mU\mU^{\rm H} - \vrho^\star}{\mZ\mU^{\rm H}} - \frac{D}{K}\innerprod{\calA^*\calA(\mU\mU^{\rm H} - \vrho^\star)}{\mZ\mU^{\rm H}} + \frac{D}{K}\innerprod{\calA^*(\ve)}{\mZ\mU^{\rm H}}\nonumber\\
  \Longrightarrow
  &&\!\!\!\!\!\!\!\!\!\!\abs{\innerprod{\mU\mU^{\rm H} - \vrho^\star}{\mZ\mU^{\rm H}}} \le 2\delta_{2r} \norm{\mU\mU^{\rm H} - \vrho^\star}{F}\norm{\mZ\mU^{\rm H}}{F} + \frac{D}{K}\abs{\innerprod{\calA^*(\ve)}{\mZ\mU^{\rm H}} } \nonumber\\
\Longrightarrow &&\!\!\!\!\!\!\!\!\!\! \norm{(\mU\mU^{\rm H} - \vrho^\star)\calP_{\mU}}{F}^2 \le 2\delta_{2r} \norm{\mU\mU^{\rm H} - \vrho^\star}{F}\norm{(\mU\mU^{\rm H} - \vrho^\star)\calP_{\mU}}{F} + \frac{D\sqrt{2r}}{K}\|\calA^*(\ve)\| \norm{(\mU\mU^{\rm H} - \vrho^\star)\calP_{\mU}}{F}\nonumber\\
  \Longrightarrow &&\!\!\!\!\!\!\!\!\!\!
  \norm{(\mU\mU^{\rm H} - \vrho^\star)\calP_{\mU}}{F} \le 2\delta_{2r} \norm{\mU\mU^{\rm H} - \vrho^\star}{F} + \frac{D\sqrt{2r}}{K}\|\calA^*(\ve)\|,
\end{eqnarray}
where $\calA^*$ is the adjoint operator of $\calA$ and is defined as $\calA^*({\bm x})=\sum_{k=1}^K x_k\mA_k$ in the first equation, $\ve$ is defined in \eqref{The measurement error in sec 3}, the first inequality follows \Cref{RIP CONDITION FRO THE Matrix SENSING OTHER PROPERTY}, and the second inequality uses  the substitutions $\mZ = (\mU\mU^{\rm H} - \vrho^\star)\mU^{{\rm H}\dagger}$ and $\abs{\innerprod{\calA^*(\ve)}{\mZ\mU^{\rm H}} }\leq \|\calA^*(\ve)\| \norm{(\mU\mU^{\rm H} - \vrho^\star)\calP_{\mU}}{*}\leq \sqrt{2r}\|\calA^*(\ve)\|\norm{(\mU\mU^{\rm H} - \vrho^\star)\calP_{\mU}}{F}$ according to the Von Neumann trace inequality.

To further characterize the locations of critical points, we need to leverage second derivative information to determine their signatures. By computing directional second derivatives, we can obtain the Wirtinger Hessian quadrature form for $\widetilde{\mDelta} = \begin{bmatrix} \mDelta \\ \mDelta^*  \end{bmatrix}\in\C^{2D\times r}$ as follows:
\begin{eqnarray}
    \label{Hessian form}
    &\!\!\!\!\!\!\!\!&\nabla^2 f(\mU,\mU^*)[\widetilde{\mDelta},\widetilde{\mDelta}]\nonumber\\
    &\!\!\!\!=\!\!\!\!&\begin{bmatrix} (\text{vec}(\mDelta))^{\rm H}  & (\text{vec}(\mDelta^*))^{\rm H}  \end{bmatrix}\begin{bmatrix}
    \frac{\partial^2 f(\mU,\mU^*)}{\partial{\text{vec}(\mU^*)}\partial{\text{vec}(\mU)^\top}} \frac{\partial^2 f(\mU,\mU^*)}{\partial{\text{vec}(\mU^*)}\partial{\text{vec}(\mU^*)^\top}} \\  \frac{\partial^2 f(\mU,\mU^*)}{\partial{\text{vec}(\mU)}\partial{\text{vec}(\mU)^\top}} \frac{\partial^2 f(\mU,\mU^*)}{\partial{\text{vec}(\mU)}\partial{\text{vec}(\mU^*)^\top}}
    \end{bmatrix}\begin{bmatrix} \text{vec}(\mDelta) \\ \text{vec}(\mDelta^*)  \end{bmatrix}\nonumber\\
    &\!\!\!\!=\!\!\!\!& \bigg\<\mDelta, \lim_{t\to 0}\frac{\nabla_{\mU^*} f(\mU+t\mDelta, \mU^*) - \nabla_{\mU^*} f(\mU,\mU^*)}{t} \bigg\> + \bigg\<\mDelta, \lim_{t\to 0}\frac{\nabla_{\mU^*} f(\mU, \mU^*+t\mDelta^*) - \nabla_{\mU^*} f(\mU,\mU^*)}{t} \bigg\>\nonumber\\
    &\!\!\!\!\!\!\!\!&+\bigg\<\mDelta^*, \lim_{t\to 0}\frac{\nabla_{\mU} f(\mU+t\mDelta, \mU^*) - \nabla_{\mU} f(\mU,\mU^*)}{t} \bigg\> + \bigg\<\mDelta^*, \lim_{t\to 0}\frac{\nabla_{\mU} f(\mU, \mU^*+t\mDelta^*) - \nabla_{\mU} f(\mU,\mU^*)}{t} \bigg\>\nonumber\\
    &\!\!\!\!=\!\!\!\!&\frac{2D}{K}\sum_{k=1}^K|\<\mA_k, \mU\mDelta^{\rm H} \>|^2 + \frac{D}{K}\sum_{k=1}^K\<\mA_k, \mU\mDelta^{\rm H} \>\<\mA_k, \mU\mDelta^{\rm H} \> + \frac{D}{K}\sum_{k=1}^K\<\mA_k, \mU\mDelta^{\rm H} \>^*\<\mA_k, \mU\mDelta^{\rm H} \>^*\nonumber\\
    &\!\!\!\!\!\!\!\!& + \frac{2D}{K}\innerprod{\calA^*(\calA(\mU\mU^{\rm H}) - \wh \vy) }{\mDelta\mDelta^{\rm H}},
\end{eqnarray}
where we utilize $\mA_k = \mA_k^{\rm H}$ and $\<\mA_k, \mU\mDelta^{\rm H} \> = \<\mA_k, \mDelta\mU^{\rm H} \>^*$ in the last line.

With \eqref{The first order information}, the Wirtinger Hessian at any critical point $\mU$ along the direction $\widetilde{\mDelta} = \mU - \mU^\star$ \footnote{To avoid carrying $\mR$ in our equations, we perform the change of variable $\mU^\star \gets \mU^\star\mR$.} is given by
\begin{eqnarray}
\label{Second order condition in the global optimality}
&\!\!\!\!\!\!\!\!& \nabla^2 f(\mU,\mU^*)[\widetilde{\mDelta},\widetilde{\mDelta}] \nonumber\\
&\!\!\!\!\leq \!\!\!\!& \frac{2D}{K}\innerprod{\calA^*(\calA(\mU\mU^{\rm H}) - \wh \vy) }{\mDelta\mDelta^{\rm H}}  + \frac{4D}{K}\|\calA(\mU\mDelta^{\rm H} )\|_2^2 \nonumber\\
&\!\!\!\!= \!\!\!\!& \frac{2D}{K}\innerprod{\calA^*(\calA(\mU\mU^{\rm H}) - \wh \vy) }{(\mU - \mU^\star)(\mU - \mU^\star)^{\rm H}} + \frac{4D}{K} \norm{\calA(\mU(\mU - \mU^\star)^{\rm H})}{2}^2 \nonumber\\
&\!\!\!\!= \!\!\!\!& \frac{4D}{K} \norm{\calA(\mU(\mU - \mU^\star)^{\rm H})}{2}^2 - \frac{2D}{K}\innerprod{\calA^*(\calA(\mU\mU^{\rm H}) - \wh \vy) }{\mU\mU^{\rm H} - \vrho^\star} \nonumber\\
&\!\!\!\!= \!\!\!\!& \frac{4D}{K} \norm{\calA(\mU(\mU - \mU^\star)^{\rm H})}{2}^2 - \frac{2D}{K}\norm{\calA(\mU\mU^{\rm H} - \vrho^\star)}{2}^2 + \frac{2D}{K}\innerprod{\calA^*(\ve)  }{\mU\mU^{\rm H} - \vrho^\star}\nonumber\\
&\!\!\!\!\leq \!\!\!\!& 4(1+\delta_{2r})\norm{\mU(\mU - \mU^\star)^{\rm H}}{F}^2 - 2(1-\delta_{2r})\norm{\mU\mU^{\rm H} - \vrho^\star}{F}^2 + \frac{2\sqrt{2}D\sqrt{r}}{K}\|\calA^*(\ve)\| \norm{\mU\mU^{\rm H} - \vrho^\star}{F}\nonumber\\
&\!\!\!\!\leq \!\!\!\!& 4(1+\delta_{2r})\bracket{  \frac{1}{8}\norm{\mU\mU^{\rm H} - \vrho^\star}{F}^2 + \bracket{3 + \frac{1}{2(\sqrt{2}-1)}} \norm{(\mU\mU^{\rm H} - \vrho^\star)\calP_{\mU}}{F}^2  } \nonumber\\
&\!\!\!\!\!\!\!\!& - 2(1-\delta_{2r})\norm{\mU\mU^{\rm H} - \vrho^\star}{F}^2 + \frac{2\sqrt{2}D\sqrt{r}}{K}\|\calA^*(\ve)\| \norm{\mU\mU^{\rm H} - \vrho^\star}{F} \nonumber\\
&\!\!\!\!\leq \!\!\!\!& - \bracket{\frac{3}{2} - \frac{5}{2}\delta_{2r} - 32(1+\delta_{2r})\delta_{2r}^2 \bracket{3 + \frac{1}{2(\sqrt{2}-1)}} } \norm{\mU\mU^{\rm H} - \vrho^\star}{F}^2 \nonumber\\
&\!\!\!\!\!\!\!\!& + \frac{16D^2r}{K^2}(1+\delta_{2r}) \bracket{3 + \frac{1}{2(\sqrt{2}-1)}}\|\calA^*(\ve)\|^2 + \frac{2\sqrt{2}D\sqrt{r}}{K}\|\calA^*(\ve)\| \norm{\mU\mU^{\rm H} - \vrho^\star}{F}\nonumber\\
&\!\!\!\!\leq \!\!\!\!& - (1.5 - 15.7\delta_{2r}  ) \norm{\mU\mU^{\rm H} - \vrho^\star}{F}^2  + \frac{53.7 D^2r}{K^2}(1+\delta_{2r}) \|\calA^*(\ve)\|^2 + \frac{2\sqrt{2}D\sqrt{r}}{K}\|\calA^*(\ve)\| \norm{\mU\mU^{\rm H} - \vrho^\star}{F},\nonumber\\
\end{eqnarray}
where the first inequality follows  $\<\mA_k, \mU\mDelta^{\rm H} \>\<\mA_k, \mU\mDelta^{\rm H} \> + \<\mA_k, \mU\mDelta^{\rm H} \>^*\<\mA_k, \mU\mDelta^{\rm H} \>^* \leq 2|\<\mA_k, \mU\mDelta^{\rm H} \>|^2$, the second equality utilizes the first-order optimality condition $\nabla_{\mU^*} f(\mU,\mU^*) = 0$, the second inequality  applies \Cref{The requirement in K of Pauli RIP}, the third inequality uses \Cref{lem:bound WDelta}, the fourth inequality uses \eqref{The first order information}, and the last line uses the assumption  $\delta_{2r}\le 0.09$.

Now the right hand side of  \eqref{Second order condition in the global optimality} is negative when
\begin{eqnarray}
\label{condition of strict saddle points}
\norm{\mU\mU^{\rm H} - \vrho^\star}{F} \geq \underbrace{\frac{\sqrt{2}+\sqrt{82.55- 762.54\delta_{2r}- 843.09\delta_{2r}^2}}{1.5-15.7\delta_{2r}}}_{h(\delta_{2r})} \frac{D\sqrt{r}}{K} \|\calA^*(\ve)\|,
\end{eqnarray}
ensuring that $\lambda_{\min}(\nabla^2 f(\mU,\mU^*))<0$ and hence these critical points are strict saddle points.

In other words, when both the first- and second-order optimality conditions are satisfied at $\mU$, the Wirtinger Hessian matrix is PSD and the right hand side of  \eqref{Second order condition in the global optimality} is non-negative, implying that $\mU$ satisfies
\begin{eqnarray}
\label{condition of all second-order critical points}
\norm{\mU\mU^{\rm H} - \vrho^\star}{F} \leq h(\delta_{2r}) \frac{D\sqrt{r}}{K} \|\calA^*(\ve)\|.
\end{eqnarray}
\end{proof}

\section{Proof of \Cref{upper bound A* e}}
\label{error bound of A*(e)}

\begin{proof}
To bound
\begin{eqnarray}
\label{expansion of A*(e)}
 \|\calA^*(\ve)\| =  \bigg\|\sum_{k=1}^K e_k \mA_k\bigg\| =  \bigg\|\sum_{k=1}^K \parans{\sum_{j=1}^D \alpha_{k,j} \wh p_{k,j} - \innerprod{\mA_k}{\vrho^\star} } \mA_k \bigg\|,
\end{eqnarray}
we first utilize $\E[f_{k,j}] = Mp_{k,j}$ to get
\begin{eqnarray}
\label{first condition of matrix Bernstein inequality}
 &\!\!\!\!\!\!\!\!&\E\bigg[{\bigg(\sum_{j=1}^D \alpha_{k,j} \wh p_{k,j} - \innerprod{\mA_k}{\vrho^\star}\bigg)\mA_k}\bigg]\nonumber\\
 &\!\!\!\!=\!\!\!\!&\E_{f_{k,j}}\bigg[{\bigg(\sum_{j=1}^D \alpha_{k,j} \frac{f_{k,j}}{M} - \sum_{j=1}^D \alpha_{k,j} p_{k,j} \bigg)\mA_k} \bigg|\mA_k\bigg]\E[\mA_k]  = {\bm 0}.
\end{eqnarray}
It follows from $\|\mA_k\| = 1$ and $\abs{ \sum_{j=1}^D \alpha_{k,j} \wh p_{k,j} - \innerprod{\mA_k}{\vrho^\star} } \le 2$ that
\begin{eqnarray}
\label{second condition of matrix Bernstein inequality}
 \bigg\|\bigg(\sum_{j=1}^D \alpha_{k,j} \wh p_{k,j} - \innerprod{\mA_k}{\vrho^\star}\bigg)\mA_k \bigg\|\leq 2.
\end{eqnarray}
Moreover, with $\E[\mA_k^2] = \mId_{D}$ and $\E\left[ \parans{ \sum_{j=1}^D \alpha_{k,j} \wh p_{k,j} - \innerprod{\mA_k}{\vrho^\star} }^2\right]  \le \frac{2}{M}$ which follows \Cref{lem:MSE-empricial-pmfs}, we have
\begin{eqnarray}
\label{third condition of matrix Bernstein inequality}
    &&\bigg\| \sum_{k=1}^K \E\left[ \parans{ \parans{\sum_{j=1}^D \alpha_{k,j} \wh p_{k,j} - \innerprod{\mA_k}{\vrho^\star} } \mA_k }^2\right] \bigg\|\nonumber\\
    &\!\!\!\! \le\!\!\!\!& \sum_{k=1}^K \E_{\wh p_{k,j}}\bigg[ \parans{ \sum_{j=1}^D \alpha_{k,j} \wh p_{k,j} - \innerprod{\mA_k}{\vrho^\star} }^2 \bigg|\mA_k \bigg] \|\E[\mA_k^2]\|  \le  \frac{2K}{M}.
\end{eqnarray}

Now, plugging \eqref{first condition of matrix Bernstein inequality}, \eqref{second condition of matrix Bernstein inequality} and \eqref{third condition of matrix Bernstein inequality} into the matrix Bernstein's inequality in \Cref{lem:Bernstein} gives that
\begin{eqnarray}
\label{matrix Bernstein's inequality}
\P{ \bigg\|\sum_{k=1}^K \parans{\sum_{j=1}^D \alpha_{k,j} \wh p_{k,j} - \innerprod{\mA_k}{\vrho^\star} } \mA_k  \bigg\| \ge t } \le D \exp\bracket{ \frac{-t^2}{4K/M + \frac{4}{3}t }  }.
\end{eqnarray}

Using $t = c\sqrt{K\log D/M}$ with a constant $c$ and $M \leq O(\frac{K}{n})$ in the above inequality leads to
\begin{eqnarray}
\label{matrix Bernstein's inequality conclusion}
&\!\!\!\!\!\!\!\!&\P{ \|\calA^*(\ve)\|  \ge c\sqrt{K\log D/M} }\nonumber\\
&\!\!\!\!=\!\!\!\!&\P{ \bigg\|\sum_{k=1}^K \parans{\sum_{j=1}^D \alpha_{k,j} \wh p_{k,j} - \innerprod{\mA_k}{\vrho^\star} } \mA_k \bigg\| \ge c\sqrt{K\log D/M} }\nonumber\\
&\!\!\!\! \le \!\!\!\!& D \exp\bracket{ \frac{-c^2\log D}{4 + \frac{4}{3}c \sqrt{\frac{M\log D}{K}} }  }.
\end{eqnarray}
Hence, we have $\|\calA^*(\ve)\|  \leq c\sqrt{K\log D/M}$ with probability $1 - e^{-\Omega(n /(1 + \sqrt{Mn/K}) )}$.
\end{proof}

\section{Proof of \Cref{lem:RegularityCondition_f}}
\label{Proof_LocalConvergence_Conclusion 1}

\begin{proof}
We first define one loss function  $F(\mU, \mU^*)=\frac{1}{2}\|\mU\mU^{\rm H}-\mU^\star{\mU^\star}^{\rm H}\|_F^2$. Then we analyze
\begin{eqnarray}
\label{fMinusF_equation}
&\!\!\!\!\!\!\!\!&\Re{\bigg\<\nabla_{\mU^*} F(\mU,\mU^*)-\nabla_{\mU^*} f(\mU,\mU^*), \mU-\mU^\star\mR \bigg\>}\nonumber\\
&\!\!\!\!=\!\!\!\!& \Re{\bigg\<(\mU\mU^{\rm H}-\mU^\star{\mU^\star}^{\rm H})\mU-\frac{D}{K}\sum_{k=1}^K\<\mA_k, \mU\mU^{\rm H}-\mU^\star{\mU^\star}^{\rm H} \>\mA_k \mU+\frac{D}{K}\sum_{k=1}^K e_k\mA_k\mU, \mU-\mU^\star\mR \bigg\>}\nonumber\\
&\!\!\!\!=\!\!\!\!& \Re{\<(\mU\mU^{\rm H}-\mU^\star{\mU^\star}^{\rm H}), (\mU-\mU^\star\mR)\mU^{\rm H}  \>} - \frac{D}{K}\Re{\<\calA(\mU\mU^{\rm H}-\mU^\star{\mU^\star}^{\rm H}), \calA((\mU-\mU^\star\mR)\mU^{\rm H})  \> }\nonumber\\
&\!\!\!\!\!\!\!\!&+\Re{\bigg\<  \frac{D}{K}\sum_{k=1}^K e_k\mA_k, (\mU-\mU^\star\mR)\mU^{\rm H} \bigg\>}\nonumber\\
&\!\!\!\!\leq\!\!\!\!& 2\delta_{3r}\|\mU\mU^{\rm H}-\mU^\star{\mU^\star}^{\rm H}\|_F\|(\mU-\mU^\star\mR)\mU^{\rm H}\|_F+\frac{D\sqrt{r}}{K}\|(\mU-\mU^\star\mR)\mU^{\rm H}\|_F\|\calA^*(\ve)\|\nonumber\\
&\!\!\!\!=\!\!\!\!&\delta_{3r}\|\mU\mU^{\rm H}-\mU^\star{\mU^\star}^{\rm H}\|_F^2+(\frac{1}{40}+\delta_{3r})\|(\mU-\mU^\star\mR)\mU^{\rm H}\|_F^2+\frac{10D^2r}{K^2}\|\calA^*(\ve)\|^2,
\end{eqnarray}
where the first inequality uses \Cref{RIP CONDITION FRO THE Matrix SENSING OTHER PROPERTY} and $\<\sum_{k=1}^K e_k\mA_k\mU, \mU-\mU^\star\mR \>\leq\|\calA^*(\ve)\|\|(\mU-\mU^\star\mR)\mU^{\rm H}\|_* \leq\sqrt{r}\|(\mU-\mU^\star\mR)\mU^{\rm H}\|_F\|\calA^*(\ve)\|$ based on the Von Neumann trace inequality.

Moreover, we can also derive
\begin{eqnarray}
\label{f_FMinusF_F_equation}
&\!\!\!\!\!\!\!\!&\|\nabla_{\mU^*} f(\mU,\mU^*)\|_F-\|\nabla_{\mU^*} F(\mU,\mU^*)\|_F\nonumber\\
&\!\!\!\!\leq\!\!\!\!& \|\nabla_{\mU^*} f(\mU,\mU^*) - \nabla_{\mU^*} F(\mU,\mU^*)\|_F\nonumber\\
&\!\!\!\!=\!\!\!\!&\bigg\|\frac{D}{K}\sum_{k=1}^K\<\mA_k, \mU\mU^{\rm H}-\mU^\star{\mU^\star}^{\rm H} \>\mA_k \mU-\frac{D}{K}\sum_{k=1}^K e_k\mA_k\mU- (\mU\mU^{\rm H}-\mU^\star{\mU^\star}^{\rm H})\mU \bigg\|_F\nonumber\\
&\!\!\!\!\leq\!\!\!\!&\bigg\|\frac{D}{K}\sum_{k=1}^K\<\mA_k, \mU\mU^{\rm H}-\mU^\star{\mU^\star}^{\rm H} \>\mA_k \mU-(\mU\mU^{\rm H}-\mU^\star{\mU^\star}^{\rm H})\mU\bigg\|_F+\bigg\|\frac{D}{K}\sum_{k=1}^K e_k\mA_k\mU\bigg\|_F\nonumber\\
&\!\!\!\!\leq\!\!\!\!&\frac{5\delta_{3r}}{2}\|\mU\mU^{\rm H}-\mU^\star{\mU^\star}^{\rm H}\|_F+\frac{5D}{4K}\|\calA^*(\ve)\|,
\end{eqnarray}
where the last line follows \Cref{RIP CONDITION FRO THE Matrix SENSING OTHER PROPERTY} and $\|\mU\|_F\leq\frac{5}{4}$ due to $\|\mU-\mU^\star\mR \|_F\leq\frac{\sigma_r(\mU^\star)}{4}$ and $\|\mU^\star\|_F=1$.

Considering that $\|\nabla_{\mU^*} F(\mU,\mU^*)\|_F\leq\frac{5}{4}\|\mU\mU^{\rm H}-\mU^\star{\mU^\star}^{\rm H}\|_F$, we have
\begin{eqnarray}
\label{BoundOf_f_square}
\|\nabla_{\mU^*} f(\mU,\mU^*)\|_F^2\leq \frac{25}{8}(1+2\delta_{3r})^2\|\mU\mU^{\rm H}-\mU^\star{\mU^\star}^{\rm H}\|_F^2+\frac{25D^2}{8K^2}\|\calA^*(\ve)\|^2.
\end{eqnarray}

Combining \Cref{lem:RegularityCondition_F} and \eqref{fMinusF_equation}, we finally have
\begin{eqnarray}
\label{RegularityCondition_Result}
&\!\!\!\!\!\!\!\!&\Re{\<\nabla_{\mU^*} f(\mU,\mU^*), \mU-\mU^\star\mR \>}\nonumber\\
&\!\!\!\!=\!\!\!\!&\Re{\<\nabla_{\mU^*} f(\mU,\mU^*) - \nabla_{\mU^*} F(\mU,\mU^*), \mU-\mU^\star\mR \>}+\Re{\<\nabla_{\mU^*} F(\mU,\mU^*), \mU-\mU^\star\mR \>} \nonumber\\
&\!\!\!\!\geq\!\!\!\!&-\delta_{3r}\|\mU\mU^{\rm H}-\mU^\star{\mU^\star}^{\rm H}\|_F^2-\bigg(\frac{1}{40}
+\delta_{3r}\bigg)\|(\mU-\mU^\star\mR)\mU^{\rm H}\|_F^2-\frac{10D^2r}{K^2}\|\calA^*(\ve)\|^2\nonumber\\
&\!\!\!\!\!\!\!\!&+\frac{1}{4}\|\mU\mU^{\rm H}-\mU^\star{\mU^\star}^{\rm H} \|_F^2+\frac{1}{20}\|(\mU-\mU^\star\mR)\mU^{\rm H} \|_F^2+\frac{\sigma_r^2(\mU^\star)}{4}\|\mU-\mU^\star\mR \|_F^2\nonumber\\
&\!\!\!\!\geq\!\!\!\!& \bigg(\frac{1}{4} - \delta_{3r}\bigg)\|\mU\mU^{\rm H}-\mU^\star{\mU^\star}^{\rm H} \|_F^2+\frac{\sigma_r^2(\mU^\star)}{4}\|\mU-\mU^\star\mR \|_F^2-\frac{10D^2r}{K^2}\|\calA^*(\ve)\|^2\nonumber\\
&\!\!\!\!\geq\!\!\!\!&\frac{\sigma_r^2(\mU^\star)}{4}\|\mU-\mU^\star\mR \|_F^2+\frac{2(1-4\delta_{3r})}{25(1+2\delta_{3r})^2}\|\nabla_{\mU^*} f(\mU,\mU^*)\|_F^2-\frac{11D^2r}{K^2}\|\calA^*(\ve)\|^2,
\end{eqnarray}
where the second and last inequalities respectively follow $\delta_{3r}\leq\frac{1}{40}$ and \eqref{BoundOf_f_square}.

\end{proof}

\section{Proof of \Cref{The spectral initialization in main paper}}
\label{Proof_Upper bound of noisy spectral initialization}

\begin{proof}
Following \cite{zhang2021preconditioned}, we begin by introducing the restricted Frobenius norm, defined as
\begin{eqnarray}
\label{The defi of restricted F norm}
  \|\mX \|_{F,r} = \max_{\mY\in\C^{m\times n}, \ \|\mY\|_F\leq 1, \atop \text{rank}(\mY) = r}|\<\mX, \mY \>|
\end{eqnarray}
for any $\mX\in\C^{m\times n}$.
Next, we analyze
\begin{eqnarray}
    \label{Upper bound of noisy spectral initialization proof}
    \|\vrho_0 - \vrho^\star \|_F &\!\!\!\!=\!\!\!\!& \|\mU_0 \mU_0^{\rm H} - \vrho^\star\|_{F,2r}\nonumber\\
    &\!\!\!\!\leq\!\!\!\!& 2\bigg\|\frac{D}{K}\sum_{k=1}^K \wh y_k\mA_k -   \vrho^\star  \bigg\|_{F,2r}\nonumber\\
    &\!\!\!\!\leq\!\!\!\!&2 \bigg\|\frac{D}{K}\calA^*(\calA(\vrho^\star))  -  \vrho^\star\bigg\|_{F,2r} + \frac{2D}{K}\|\calA^*(\ve) \|_{2r}\nonumber\\
    &\!\!\!\!\leq\!\!\!\!& 2\max_{\mY\in\C^{D\times D}, \ \|\mY\|_F\leq 1, \atop \text{rank}(\mY) = 2r} \bigg\< \bigg(\frac{D}{K} \calA^* \calA - \mathcal{I} \bigg)(\vrho^\star), \mY \bigg\> + \frac{2D\sqrt{2r}}{K}\|\calA^*(\ve) \|\nonumber\\
    &\!\!\!\!\leq\!\!\!\!&4\delta_{3r}\|\vrho^\star\|_F + O\bigg(\sqrt{\frac{D^2(\log D)r}{KM}}\bigg),
\end{eqnarray}
where the first inequality follows from the quasi-optimality property of truncated eigenvalue decomposition projection \cite{Oseledets11} and  the last line uses \Cref{RIP CONDITION FRO THE Matrix SENSING OTHER PROPERTY} and \Cref{upper bound A* e} which are simultaneously satisfied with probability $1 - e^{-C/3} - e^{-\Omega(n /(1 + \sqrt{Mn/K}) )} $.

Furthermore, based on \Cref{The relationship between two distances}, we can obtain
\begin{eqnarray}
    \label{The final conclusion of intial conclusion analysis in appendix}
    \|\mU_0-\mU^\star\mR_0 \|_F^2 &\!\!\!\!\leq\!\!\!\!& \frac{1}{2(\sqrt{2}-1)\sigma_r^2(\mU^\star)}\|\vrho_0 - \vrho^\star \|_F^2\nonumber\\
    &\!\!\!\!\leq\!\!\!\!& \frac{\bigg(4\delta_{3r}\|\vrho^\star\|_F + O\bigg(\sqrt{\frac{D^2(\log D)r}{KM}}\bigg)\bigg)^2}{2(\sqrt{2}-1)\sigma_r^2(\mU^\star)},
\end{eqnarray}
where $\mR_0 = \argmin_{\mR'\in\calU_r}\|\mU_0 - \mU^\star\mR'\|_F$.

\end{proof}

\section{Auxiliary Materials}
\label{sec: auxiliary materials}

\begin{lemma} Given a set of probabilities $(p_1,\dots, p_D)$, we can generate empirical probabilities $\wh p_q = \frac{\# \text{ of $q$-th output} }{M}, \ \forall q \in[D]:=\{1,\ldots,D\}$ by repeating the measurement process $M$ times and taking the average of the outcomes. So $(\wh p_1,\dots, \wh p_D )$ follows the multinomial distribution $\operatorname{Multinomial}(M, \vp)$ with the covariance matrix $\mSigma$. Elements of $\mSigma$ are $\Sigma_{i,k} = \begin{cases} \frac{p_i(1-p_i)}{M}, & i=k, \\ -\frac{p_ip_k}{M}, & i\neq k. \end{cases}$.
Then
\begin{eqnarray}
\E \parans{\sum_{q=1}^D \alpha_{q} \parans{ \wh p_q - p_q} }^2  \le \frac{2-1/D}{M}
\end{eqnarray}
holds for any  $\alpha_q = \pm 1, q = 1,\ldots, D$.
\label{lem:MSE-empricial-pmfs}\end{lemma}
\begin{proof}
Based on the covariance matrix of the multinomial distribution, we have
\begin{eqnarray}
 \E\left[ (\wh p_q - p_q) (\wh p_{q'} - p_{q'})\right] =  \begin{cases} -\frac{1}{M} p_q p_{q'}, & q \neq q', \\  \frac{1}{M} p_q (1 - p_q) , & q = q', \end{cases}
\end{eqnarray}
and further derive
\begin{eqnarray}
\E \parans{\sum_{q=1}^D \alpha_{j} \parans{ \wh p_q - p_q} }^2 & \!\!\!\!=\!\!\!\! &\sum_{q=1}^D \sum_{q' = 1}^D \alpha_{q} \alpha_{q'} \E\left[ \parans{ \wh p_q - p_q }  \parans{ \wh p_{q'} - p_{q'}} \right]\nonumber\\
& \!\!\!\!\le \!\!\!\!& \sum_{q=1}^D \sum_{q' = 1}^D \abs{ \E\left[ \parans{ \wh p_q - p_q }  \parans{ \wh p_{q'} - p_{q'}} \right] }\nonumber\\
& \!\!\!\!=\!\!\!\!& \sum_{q\neq q'} \frac{1}{M} p_q p_{q'} + \sum_{q=1}^D \frac{1}{M} p_q (1 - p_q)\nonumber\\
& \!\!\!\!\le\!\!\!\!& \frac{2-2/D}{M},
\end{eqnarray}
where the last inequality follows $\sum_{q\neq q'}  p_q p_{q'} -\sum_{q=1}^D p_q^2 = 1 - 2\sum_{q=1}^D p_q^2 \leq 1 - \frac{2}{D} $ and $ \sum_{q=1}^D p_q =   1 $.

\end{proof}

\begin{lemma}
\label{RIP CONDITION FRO THE Matrix SENSING OTHER PROPERTY}
Suppose that $\calA$ obeys the $(2r,\delta_{2r})$-RIP in \Cref{thm:RIP-Pauli}. Then for any $\vrho_1,\vrho_2\in\C^{D\times D}$ of rank at most $r$, one has
\begin{eqnarray}
    \label{RIP CONDITION FRO THE Matrix SENSING OTHER PROPERTY_1}
    \bigg|\frac{D}{K}\<\calA(\vrho_1),\calA(\vrho_2)\> - \<\vrho_1,\vrho_2\>\bigg|\leq 2\delta_{2r}\|\vrho_1\|_F\|\vrho_2\|_F,
\end{eqnarray}
or equivalently,
\begin{eqnarray}
    \label{RIP CONDITION FRO THE Matrix SENSING OTHER PROPERTY_1_1}
    \bigg|\bigg\<\bigg(\frac{D}{K}\calA^*\calA - \mathcal{I}\bigg)(\vrho_1), \vrho_2 \bigg\>\bigg|\leq 2\delta_{2r}\|\vrho_1\|_F\|\vrho_2\|_F,
\end{eqnarray}
where $\calA^*$ is the adjoint operator of $\calA$ and is defined as $\calA^*({\bm x})=\sum_{k=1}^K x_k\mA_k$.

In addition, we also have
\begin{eqnarray}
    \label{RIP CONDITION FRO THE Matrix SENSING OTHER PROPERTY_2}
    \bigg\|\frac{D}{K}\sum_{k=1}^K\<\mA_k,\vrho\>\mA_k\mB- \vrho\mB\bigg\|_F\leq 2\delta_{2r}\|\vrho\|_F\|\mB\|,
\end{eqnarray}
where $\vrho\in\C^{D\times D}$ and $\mB\in\C^{D\times r}$.
\end{lemma}

\begin{proof}
To conveniently analyze \Cref{RIP CONDITION FRO THE Matrix SENSING OTHER PROPERTY}. we assume that $\|\vrho_1\|_F=\|\vrho_2\|_F=1$ for $\vrho_1,\vrho_2\in\C^{D\times D}$.
We can expand $|\<\calA(\vrho_1), \calA(\vrho_2)  \>|$ as follows:
\begin{eqnarray}
    \label{Proof of RIP other properties 1}
    &\!\!\!\!\!\!\!\!&\bigg|\frac{D}{K}\<\calA(\vrho_1), \calA(\vrho_2)  \> - \< \vrho_1,\vrho_2 \>\bigg| \nonumber\\
    &\!\!\!\!=\!\!\!\!& \frac{1}{4}\bigg|\frac{D}{K}\bigg(\|\calA(\vrho_1 + \vrho_2)\|_2^2 - \|\calA(\vrho_1 - \vrho_2)\|_2^2 - i\|\calA(\vrho_1+i\vrho_2)\|_2^2 + i\|\calA(\vrho_1 - i\vrho_2)\|_2^2 \bigg)\nonumber\\
    &\!\!\!\!\!\!\!\!& - \bigg( \|\vrho_1 + \vrho_2\|_F^2 - \|\vrho_1 - \vrho_2\|_F^2 - i\|\vrho_1+i\vrho_2\|_F^2 + i\|\vrho_1 - i\vrho_2\|_F^2  \bigg)  \bigg|\nonumber\\
    &\!\!\!\!\leq \!\!\!\!& \frac{1}{4} \bigg( \bigg| \frac{D}{K}\|\calA(\vrho_1 + \vrho_2)\|_2^2 - \|\vrho_1 + \vrho_2\|_F^2 \bigg| + \bigg| \frac{D}{K}\|\calA(\vrho_1 - \vrho_2)\|_2^2 - \|\vrho_1 - \vrho_2\|_F^2  \bigg|  \nonumber\\
    &\!\!\!\! \!\!\!\!&  + \bigg|\frac{D}{K}\calA(\vrho_1+i\vrho_2)\|_2^2  - \|\vrho_1+i\vrho_2\|_F^2 \bigg| +  \bigg|\frac{D}{K}\|\calA(\vrho_1 - i\vrho_2)\|_2^2 -  \|\vrho_1 - i\vrho_2\|_F^2  \bigg| \bigg)\nonumber\\
    &\!\!\!\!\leq \!\!\!\!& \frac{\delta_{2r}}{4} (\|\vrho_1 + \vrho_2\|_F^2  + \|\vrho_1 - \vrho_2\|_F^2 + \|\vrho_1+i\vrho_2\|_F^2 +\|\vrho_1 - i\vrho_2\|_F^2  ) \nonumber\\
    &\!\!\!\! =  \!\!\!\!& 2\delta_{2r}.
\end{eqnarray}

Now, we can obtain
\begin{eqnarray}
    \label{Proof of RIP other properties 2}
\bigg|\frac{D}{K}\<\calA(\vrho_1),\calA(\vrho_2)\> - \<\vrho_1,\vrho_2\>\bigg|\leq 2\delta_{2r}\|\vrho_1\|_F\|\vrho_2\|_F.
\end{eqnarray}

Furthermore, by referring to \cite{Ma21TSP} and utilizing equation \eqref{Proof of RIP other properties 2}, we can directly obtain the following result:
\begin{eqnarray}
    \label{RIP CONDITION FRO THE Matrix SENSING OTHER PROPERTY_2_proof}
    \bigg\|\frac{D}{K}\sum_{k=1}^K\<\mA_k,\vrho\>\mA_k\mB- \vrho\mB\bigg\|_F\leq 2\delta_{2r}\|\vrho\|_F\|\mB\|,
\end{eqnarray}
where $\vrho\in\C^{D\times D}$ and $\mB\in\C^{D\times r}$.

\end{proof}

\begin{lemma}
\label{Two distances relationship}
Let $\mU, \mY\in\C^{D\times r}$. Additionally, let $\mU^{\rm H}\mY = \mY^{\rm H}\mU$ be a PSD matrix. Then, we have
\begin{eqnarray}
    \label{Two distances relationship1}
\|(\mU - \mY)\mU^{\rm H}\|_F^2\leq \frac{1}{2(\sqrt{2} -1)}\|\mU\mU^{\rm H} - \mY\mY^{\rm H}\|_F^2.
\end{eqnarray}
\end{lemma}
\begin{proof}
First we define $\mDelta = \mU - \mY$  and then expand the right hand side in \eqref{Two distances relationship1}.
\begin{eqnarray}
    \label{Two distances relationship2}
    &\!\!\!\!\!\!\!\!&\|\mU\mU^{\rm H} - \mY\mY^{\rm H}\|_F^2 \nonumber\\
    &\!\!\!\!=\!\!\!\!& \|\mU\mDelta^{\rm H} + \mDelta \mU^{\rm H} -\mDelta \mDelta^{\rm H}\|_F^2 \nonumber\\
    &\!\!\!\!=\!\!\!\!&\trace(\mDelta\mU^{\rm H}\mU\mDelta^{\rm H} + \mU\mDelta^{\rm H}\mDelta\mU^{\rm H}+\mDelta\mDelta^{\rm H}\mDelta\mDelta^{\rm H} +2\mDelta\mU^{\rm H}\mDelta\mU^{\rm H} \nonumber\\
    &\!\!\!\!\!\!\!\!&- 2\mDelta\mDelta^{\rm H}\mDelta\mU^{\rm H}-2\mDelta\mDelta^{\rm H}\mU\mDelta^{\rm H})\nonumber\\
    &\!\!\!\!=\!\!\!\!&2\trace(\mU^{\rm H}\mU\mDelta^{\rm H}\mDelta) +\|\mDelta^{\rm H}\mDelta\|_F^2 +\|\sqrt{2}\mU^{\rm H}\mDelta\|_F^2-2\sqrt{2}\Re{\<\mDelta^{\rm H}\mDelta,  \mU^{\rm H}\mDelta \>}\nonumber\\
    &\!\!\!\!\!\!\!\!&+2\sqrt{2}\Re{\<\mDelta^{\rm H}\mDelta,  \mU^{\rm H}\mDelta \>} -4\Re{\trace(\mDelta\mDelta^{\rm H}\mDelta\mU^{\rm H})}\nonumber\\
    &\!\!\!\!=\!\!\!\!&2\trace(\mU^{\rm H}\mU\mDelta^{\rm H}\mDelta) + \|\mDelta^{\rm H}\mDelta - \sqrt{2}\mU^{\rm H}\mDelta\|_F^2 - 2(2-\sqrt{2})\Re{\trace(\mU^{\rm H}\mDelta\mDelta^{\rm H}\mDelta)}\nonumber\\
    &\!\!\!\!\geq\!\!\!\!&2\trace(\mU^{\rm H}\mU\mDelta^{\rm H}\mDelta)  - 2(2-\sqrt{2})\Re{\trace(\mU^{\rm H}\mDelta\mDelta^{\rm H}\mDelta)}\nonumber\\
    &\!\!\!\!=\!\!\!\!& 2\Re{\trace(((\sqrt{2} - 1)\mU^{\rm H}\mU + (2-\sqrt{2})\mU^{\rm H}\mY )\mDelta^{\rm H}\mDelta)}\nonumber\\
    &\!\!\!\!\geq\!\!\!\!&2(\sqrt{2} - 1)\Re{\trace(\mU^{\rm H}\mU\mDelta^{\rm H}\mDelta)}\nonumber\\
    &\!\!\!\!=\!\!\!\!&2(\sqrt{2} - 1)\|(\mU - \mY)\mU^{\rm H}\|_F^2,
\end{eqnarray}
where the third equation follows $\mDelta^{\rm H}\mU = (\mU - \mY)^{\rm H}\mU = \mU^{\rm H}\mU - \mU^{\rm H}\mY = \mU^{\rm H}\mDelta$. In addition, we use $\trace(\mA\mB)\geq 0$ for PSD matrices $\mA$ and $\mB$ in the second inequality.
\end{proof}

Following \Cref{Two distances relationship}, we can directly extend the result from the real case in \cite[Lemma 3.6]{li2019non} to the complex case.
\begin{lemma}\label{lem:bound WDelta}
For any matrices $\mC,\mD \in\C^{D\times r}$, let $\mP_{\mC}$ be the orthogonal projector onto the range of $\mC$. Let $\mR = \argmin_{\mR'\in\calU_r}\|\mC - \mD\mR'\|_F$ where $\calU_r$ denotes the unitary matrix of size $r$. Then, we have
\begin{eqnarray}
	\|\mC\left(\mC - \mD\mR\right)^{\rm H} \|_F^2 \leq \frac{1}{8}\|\mC\mC^\H - \mD\mD^\H \|_F^2 + \bracket{3+\frac{1}{2(\sqrt{2} - 1)}} \|(\mC\mC^\H - \mD\mD^\H)\mP_{\mC}\|_F^2.
\end{eqnarray}
\end{lemma}

\begin{lemma}(Matrix Bernstein's inequality \cite[Theorem 6.6.1]{tropp2015introduction}) Let $\mA_1,\ldots,\mA_n\in\C^{D\times D}$ be  independent, random, Hermitian matrices, and assume that each one is uniformly bounded:
\begin{eqnarray}
\E[\mA_q] = {\bm 0} \ \ \text{and} \ \ \|\mA_q\| \le R, \ \ \forall q = 1,\ldots, Q.
\end{eqnarray}
Then, for any $t>0$,
\begin{eqnarray}
\P{\bigg\|\sum_{q=1}^Q \mA_q\bigg\| \ge t } \le D \exp\bracket{ \frac{-t^2}{2 \|\sum_{q=1}^Q \E[\mA_q^2]\| + \frac{2R}{3}t }  }.
\end{eqnarray}
\label{lem:Bernstein}\end{lemma}

\begin{lemma} \label{lem:RegularityCondition_F}
Consider the loss function $F(\mU,\mU^*)=\frac{1}{2}\|\mU\mU^{\rm H}-\mU^\star{\mU^\star}^{\rm H}\|_F^2$ \ for $\mU, \mU^\star\C^{D\times r}$.
For all $\mU$ satisfying the condition
\begin{eqnarray}
\label{RegularityCondition_F_Requirement}
\|\mU-\mU^\star\mR \|\leq\frac{\sigma_r(\mU^\star)}{4},
\end{eqnarray}
we have the following conclusion:
\begin{eqnarray}
\label{RegularityCondition_F_Conclusion}
&\!\!\!\!\!\!\!\!&\Re{\<\nabla_{\mU^*} F(\mU,\mU^*), \mU-\mU^\star\mR \>}\nonumber\\
&\!\!\!\!\geq\!\!\!\!& \frac{1}{4}\|\mU\mU^{\rm H}-\mU^\star{\mU^\star}^{\rm H} \|_F^2+\frac{1}{20}\|(\mU-\mU^\star\mR)\mU^{\rm H} \|_F^2+\frac{\sigma_r^2(\mU^\star)}{4}\|\mU-\mU^\star\mR \|_F^2.
\end{eqnarray}
\end{lemma}

\begin{proof}
Let $\mU, \mU^\star\in\C^{D\times r}$, and define $\mDelta = \mU - \mU^\star\mR$, where $\mR$ is the unitary matrix that minimizes $\|\mU - \mU^\star\mR\|_F$. We denote $\mA\mSigma\mB^{\rm H}$ as the singular value decomposition of ${\mU^\star}^{\rm H}\mU$, and we have the optimal $\mR = \mA\mB^{\rm H}$. Consequently, we can further obtain $\mU^{\rm H}\mU^\star\mR = \mB \mSigma \mB^{\rm H} = (\mU^\star\mR)^{\rm H}\mU$, which means that $\mU^{\rm H}\mU^\star\mR$ is a Hermitian PSD matrix. Additionally, we also have $\mDelta^{\rm H}\mU^\star\mR = \mU^{\rm H}\mU^\star\mR - \mR^{\rm H}{\mU^\star}^{\rm H}\mU^\star\mR = (\mU^\star\mR)^{\rm H}\mU - (\mU^\star\mR)^{\rm H}\mU^\star\mR = (\mU^\star\mR)^{\rm H}\mDelta$. Hence, $\mDelta^{\rm H}\mU^\star\mR$ is a Hermitian matrix. To avoid carrying $\mR$ in our equations, we perform the change of variable $\mU^\star \gets \mU^\star\mR$. Without loss of generality, we assume $\mR = \mId$ and have $\mU^{\rm H}\mU^\star\succeq 0$ and $\mDelta^{\rm H}\mU^\star = {\mU^\star}^{\rm H}\mDelta$.

Based on the fact
\begin{eqnarray}
\label{factorization regualrity condition 1}
&\!\!\!\!\!\!\!\!&\Re{\<\nabla_{\mU^*} F(\mU,\mU^*), \mU-\mU^\star \>} \nonumber\\
&\!\!\!\! =\!\!\!\!& \frac{1}{2}\trace\big((\nabla_{\mU^*} F(\mU,\mU^*))^{\rm H}(\mU-\mU^\star) \big) + \frac{1}{2}\trace\big((\mU-\mU^\star)^{\rm H}\nabla_{\mU^*} F(\mU,\mU^*)\big)
\end{eqnarray}
and $\nabla_{\mU^*} F(\mU,\mU^*) =(\mU\mU^{\rm H}-\mU^\star{\mU^\star}^{\rm H})\mU$,
we will instead prove
\begin{eqnarray}
\label{factorization regualrity condition 2}
&&\frac{1}{2}\trace((\mU\mU^{\rm H} - \mU^\star{\mU^\star}^{\rm H})\mDelta\mU^{\rm H}) + \frac{1}{2}\trace(\mU\mDelta^{\rm H}(\mU\mU^{\rm H} - \mU^\star{\mU^\star}^{\rm H})  )\nonumber\\
&&-c_1\|\mU\mU^{\rm H}-\mU^\star{\mU^\star}^{\rm H}\|_F^2 -c_2 \sigma_r^2(\mU^\star)\|\mU - \mU^\star\|_F^2 - c_3\|(\mU-\mU^\star)\mU^{\rm H} \|_F^2 \geq 0,
\end{eqnarray}
where $c_1,c_2,c_3$ are positive constants.

By $\mU = \mDelta + \mU^\star$, we can expand the LHS of \eqref{factorization regualrity condition 2} as following:
\begin{eqnarray}
\label{factorization regualrity condition 3}
&\!\!\!\!\!\!\!\!&\frac{1}{2}\trace((\mU\mU^{\rm H} - \mU^\star{\mU^\star}^{\rm H})\mDelta\mU^{\rm H}) + \frac{1}{2}\trace(\mU\mDelta^{\rm H}(\mU\mU^{\rm H} - \mU^\star{\mU^\star}^{\rm H})  )\nonumber\\
&\!\!\!\!\!\!\!\!&-c_1\|\mU\mU^{\rm H}-\mU^\star{\mU^\star}^{\rm H}\|_F^2 -c_2 \sigma_r^2(\mU^\star)\|\mU - \mU^\star\|_F^2 - c_3\|(\mU-\mU^\star)\mU^{\rm H} \|_F^2\nonumber\\
&\!\!\!\!=\!\!\!\!&(1-c_1 -c_3)\trace(\mDelta\mDelta^{\rm H}\mDelta\mDelta^{\rm H}) - c_2\sigma_r^2(\mU^\star) \trace(\mDelta^{\rm H} \mDelta) + (1-2c_1 -c_3)\trace({\mU^\star}^{\rm H}\mU^\star \mDelta^{\rm H}\mDelta )\nonumber\\
&\!\!\!\!\!\!\!\!&+(\frac{3}{2}-2c_1-c_3 )\trace(\mDelta\mDelta^{\rm H}\mDelta{\mU^\star}^{\rm H}) +(\frac{3}{2}-2c_1-c_3 )\trace(\mDelta^{\rm H}\mDelta\mDelta^{\rm H}\mU^\star)\nonumber\\
&\!\!\!\!\!\!\!\!&+(\frac{1}{2} - c_1)\trace(\mDelta{\mU^\star}^{\rm H}\mDelta{\mU^\star}^{\rm H}) +(\frac{1}{2} - c_1 )\trace(\mU^\star\mDelta^{\rm H}\mU^\star\mDelta^{\rm H}) \nonumber\\
&\!\!\!\!=\!\!\!\!&(1-c_1 -c_3)\trace(\mDelta\mDelta^{\rm H}\mDelta\mDelta^{\rm H}) - c_2\sigma_r^2(\mU^\star) \trace(\mDelta^{\rm H} \mDelta) + (1-2c_1 -c_3)\trace({\mU^\star}^{\rm H}\mU^\star \mDelta^{\rm H}\mDelta )\nonumber\\
&\!\!\!\!\!\!\!\!&+(3-4c_1-2c_3)\Re{\trace(\mDelta\mDelta^{\rm H}\mDelta{\mU^\star}^{\rm H})} +(1-2c_1)\Re{\trace(\mDelta{\mU^\star}^{\rm H}\mDelta{\mU^\star}^{\rm H})}\nonumber\\
&\!\!\!\!=\!\!\!\!&(1-c_1 -c_3)\trace(\mDelta\mDelta^{\rm H}\mDelta\mDelta^{\rm H}) - c_2\sigma_r^2(\mU^\star) \trace(\mDelta^{\rm H} \mDelta) + (1-2c_1 -c_3)\|\mDelta{\mU^\star}^{\rm H}\|_F^2 \nonumber\\
&\!\!\!\!\!\!\!\!& + (3-4c_1-2c_3)\trace(\mDelta\mDelta^{\rm H}\mDelta{\mU^\star}^{\rm H}) +(1-2c_1)\|{\mU^\star}^{\rm H}\mDelta\|_F^2\nonumber\\
&\!\!\!\!=\!\!\!\!& \bigg\|\frac{3-4c_1-2c_3}{2\sqrt{1-2c_1-c_3}}\mDelta\mDelta^{\rm H} + \sqrt{1-2c_1 -c_3}\mDelta{\mU^\star}^{\rm H} \bigg\|_F^2 +(1-2c_1)\|{\mU^\star}^{\rm H}\mDelta\|_F^2\nonumber\\
&\!\!\!\!\!\!\!\!& +\bigg(1-c_1 -c_3 - \frac{(3-4c_1-2c_3)^2}{4(1-2c_1-c_3)}\bigg)\trace(\mDelta\mDelta^{\rm H}\mDelta\mDelta^{\rm H}) - c_2\sigma_r^2(\mU^\star) \trace(\mDelta^{\rm H} \mDelta)
\end{eqnarray}
where the third equation follows from the fact that $\trace(\mDelta\mDelta^{\rm H}\mDelta{\mU^\star}^{\rm H})$  and $\trace(\mDelta{\mU^\star}^{\rm H}\mDelta{\mU^\star}^{\rm H})$ are real. This is due to the property that, by using $\mDelta^{\rm H}\mU^\star = {\mU^\star}^{\rm H}\mDelta$, we can respectively show that
\begin{eqnarray}
\label{property of F 1}
(\trace(\mDelta\mDelta^{\rm H}\mDelta{\mU^\star}^{\rm H}))^* = \trace(\mDelta^* \mDelta^\top \mDelta^* {\mU^\star}^\top) = \trace( \mDelta^{\rm H}\mDelta\mDelta^{\rm H}{\mU^\star}) =  \trace(\mDelta \mDelta^{\rm H}\mDelta{\mU^\star}^{\rm H}),
\end{eqnarray}
\begin{eqnarray}
\label{property of F 2}
\trace(\mDelta{\mU^\star}^{\rm H}\mDelta{\mU^\star}^{\rm H}) = \trace(\mDelta\mDelta^{\rm H}\mU^\star{\mU^\star}^{\rm H}) = \|{\mU^\star}^{\rm H}\mDelta\|_F^2.
\end{eqnarray}

To ensure $\eqref{factorization regualrity condition 3}\geq 0$, by the fact $1-c_1 -c_3 - \frac{(3-4c_1-2c_3)^2}{4(1-2c_1-c_3)}\leq 0$ for $0\leq c_1,c_3\leq 1$, we need to guarantee that
\begin{eqnarray}
\label{condition of final conclusion}
&\!\!\!\!\!\!\!\!&(1-2c_1)\|{\mU^\star}^{\rm H}\mDelta\|_F^2+(1-c_1 -c_3 - \frac{(3-4c_1-2c_3)^2}{4(1-2c_1-c_3)})\trace(\mDelta\mDelta^{\rm H}\mDelta\mDelta^{\rm H}) - c_2\sigma_r^2(\mU^\star) \trace(\mDelta^{\rm H} \mDelta)\nonumber\\
&\!\!\!\! \geq\!\!\!\!& (1-2c_1)\sigma_r^2(\mU^\star)\|\mDelta\|_F^2 - (c_1 + c_3 + \frac{(3-4c_1-2c_3)^2}{4(1-2c_1-c_3)}-1)\|\mDelta\|^2\|\mDelta\|_F^2 - c_2\sigma_r^2(\mU^\star) \|\mDelta\|_F^2\geq 0.
\end{eqnarray}
Note that \eqref{condition of final conclusion} can be satisfied when $\|\mDelta\|^2\leq \frac{1-2c_1-c_2}{c_1 + c_3 + \frac{(3-4c_1-2c_3)^2}{4(1-2c_1-c_3)}-1}\sigma_r^2(\mU^\star)\leq\frac{9}{47}\sigma_r^2(\mU^\star)\leq \frac{\sigma_r^2(\mU^\star)}{4}$ for $c_1=c_2 =\frac{1}{4}$ and $c_3=\frac{1}{20}$.

This completes the proof.

\end{proof}

By directly extending \cite[Lemma 5.4]{Tu16} to the complex case, we have
\begin{lemma}
\label{The relationship between two distances}
For any $\mU,\mX\in\C^{D\times r}$, we have
\begin{eqnarray}
    \label{Relationship between two distances}
    \|\mU - \mX\mR \|_F^2 \leq \frac{1}{2(\sqrt{2}-1)\sigma_r^2(\mX)} \|\mU\mU^{\rm H} - \mX\mX^{\rm H} \|_F^2,
\end{eqnarray}
where  $\mR\in\calU_r$.
\end{lemma}

%

\end{document}